\documentclass[letterpaper, DIV=13]{scrartcl}

\usepackage{amssymb, amsthm, mathtools,bbm}
\usepackage[square,sort,comma,numbers]{natbib}
\usepackage[none]{hyphenat}
\newtheorem{theorem}{Theorem}
\newtheorem{corollary}[theorem]{Corollary}
\newtheorem{lemma}[theorem]{Lemma}
\newtheorem{definition}[theorem]{Definition}
\newtheorem{proposition}[theorem]{Proposition}
\newtheorem{assumption}[theorem]{Assumption}
\numberwithin{equation}{section}
\numberwithin{theorem}{section}

\newcommand{\rr}{{\mathbb{R}}}

\newcommand{\nn}{{\mathbb{N}}}

\newcommand{\ee}{{\mathbb{E}\,}}

\newcommand{\pp}{{\mathbb{P}}}

\newcommand{\tr}{{\operatorname{Tr}\,}}

\newcommand{\beq}[1]{\begin{equation} \label{#1}}
\newcommand{\eeq}{\end{equation}}

\newcommand{\arcosh}{\operatorname{arcosh}}

\begin{document}
	\addtokomafont{author}{\raggedright}

	\title{ \raggedright Generalized Random Energy Models in a Transversal Magnetic Field: \\ Free Energy and Phase Diagrams}
	
     \author{\hspace{-.075in} Chokri Manai and Simone Warzel}
     \date{\vspace{-.3in}}
	
	\maketitle

	\minisec{Abstract}
	We determine explicit variational expressions for the free energy of mean-field spin glasses in a transversal magnetic field, whose glass interaction is given by a hierarchical Gaussian potential as in Derrida's Generalized Random Energy Model (GREM), its continuous version (CREM) or the non-hierarchical GREM. The corresponding phase diagrams, which generally include glass transitions as well as magnetic transitions, are discussed. In the glass phase, the free energy is generally determined by both the parameters of the classical model and the transversal field. 

\section{Introduction}\label{sec:intro}
	Studying the fate of spin glass physics with respect to quantum effects induced by a transversal field has been a topic of continuing interest in the physics community. 
	In the past 10 years this subject received an additional boost due to its relevance as a testing ground for quantum adiabatic algorithms and for many-body localised systems (see e.g.~\cite{BFKSZ13,LPS14,AW16,Bur17}). 
	
	Ever since exact solutions of the free energy of mean-field spin glasses became available \cite{Par80},  Parisi's famous replica calculations \cite{MPV86} have been extended to approximations of the quantum free energy. 	
	Notwithstanding numerous works (see e.g.~\cite{BM80,FS86,Us86,Gold90,ONS07} and references in~\cite{SIC12}), an ultimate consensus on various aspects of quantum spin glasses  such as the quantum Sherrington-Kirkpatrick (SK) model seems to be lacking even from the physics point of view. 
	
	Although the theory of classical mean-field spin glasses became an established branch of probability \cite{Tal11,Bov06}, efforts of mathematicians in this area are so far rather limited. Crawford lay the foundations with generalising Guerra and  Toninelli's proof \cite{GT02} of the existence of the free-energy to the quantum SK model \cite{Craw07}. 
	For this model a generalisation of the high-temperature analysis of \cite{ALR87} was achieved recently in \cite{LRRS19}.   Adhikari and Brennecke \cite{AB19}  used a path-integral approach  and Parisi's formula for vector-spin models to rewrite the free energy of the quantum SK model as a rather involved variational problem in terms of infinite-dimensional path overlaps. 
	
	The main aim of this work is to derive reasonably explicit variational expressions, which allow us to determine the structure of the phase diagram, for the quantum versions of three classic hierarchical mean-field spin glasses: i)~Derrida's Generalized Random Energy Model (GREM) \cite{Der85}, ii)~its continuous version (CREM), and 
	iii)~the non-hierarchical GREM by Bolthausen and Kistler~\cite{BK06}. 
	 These models were invented so to incorporate the effects of correlations of energy levels into the oversimplified Random Energy Model (REM) \cite{Der80,Der81}. 
	 Remarkably, the GREM's and CREM's built-in ultra-metric structure is shared by Parisi's solution of the SK mean-field spin glass which received mathematical blessing through Talagrand's proof \cite{Tal06}.

	 Although the built-in ultra-metric structure and the prearrangement of species in the GREM or CREM is somewhat artificial, it is nevertheless surprising that no physics prediction exists for the quantum version of these classic hierarchical mean-field spin glasses. All the more so since Goldschmidt \cite{Gold90} already in 1990 presented his formula for the free energy of the quantum REM, which was recently confirmed through a mathematical analysis~\cite{MW20}.  This gap is closed with the present paper. What we find are formulae, which express the principle that the species decide in within the groups whether to collectively follow the transversal field or stay in their classical order. The free energy is then computed as a minimum over all group decompositions. 
	 As a mathematical technique, we dubbed this principle hierarchical peeling. It is based purely on a combination of a probabilistic-geometric decomposition of the spin-configuration space and operator-theoretic techniques, which are further developments of  ideas in~\cite{MW20,MW19}. In passing, we also generalise  basic interpolation techniques to the quantum set-up. These main new technical tools are presented in Section~\ref{sec:peeling}. 

	We start the paper with a short introduction to classical hierarchical models, for which the quantum free energy is then presented in a subsequent subsection. The introduction closes with a discussion of the non-hierarchical GREM and its quantum Parisi-type formula. The proofs of the novel quantum formulae is then postponed to Section~\ref{sec:proofmain}.

	\subsection{Classical GREM and CREM}\label{subsec:GREM}
	The GREM and CREM potential $U$ is a centered Gaussian process on the Hamming cube 
	$\mathcal{Q}_N \coloneqq \{-1,1\}^N$, whose covariance matrix is given by
	\begin{equation}\label{eq:ucrem}
	\ee[U(\pmb{\sigma}) U(\pmb{\sigma}^\prime)] = N A(q(\pmb{\sigma},\pmb{\sigma}^\prime)),
	\end{equation}
	where $A \coloneqq [0,1] \to [0,1]$ is a non-decreasing, right-continuous, and normalized function, i.e., $A(1) =1$,
	which does not depend on $N$. Moreover, $q$ denotes the  normalized lexicographic overlap of spin configurations 
	$ \pmb{\sigma}, \pmb{\sigma}' \in \mathcal{Q}_N $, i.e.
	\begin{equation}\label{eq:overlap}
	q(\pmb{\sigma},\pmb{\sigma}^\prime) \coloneqq \begin{cases}
	1 & \text{ if } \pmb{\sigma} = \pmb{\sigma}^\prime, \\
	\frac1N (\min \{1 \leq i \leq N; \sigma_i \neq \sigma^\prime_i \} -1)& \text{ else }. \end{cases} 
	\end{equation}
	The induced metric $ d(\pmb{\sigma},\pmb{\sigma}^\prime)= \ee[|U(\pmb{\sigma}) - U(\pmb{\sigma}^\prime)|^2]^{1/2}   $ on the Hamming cube is an ultrametric. \\

   In the GREM one further assumes that the distribution function $A$ is a step function with $ n \in \mathbb{N} $ jumps of height $a_k$ at the values $0 = x_0 < x_1 < x_2 < \cdots < x_n = 1$. The Gaussian potential $U$ can then be expressed in terms of independent standard Gaussian variables. To this end, one decomposes $\pmb{\sigma} \in \mathcal{Q}_N$ into $ n $ blocks ('species'), $\pmb{\sigma} = (\pmb{\sigma}_1,\ldots,\pmb{\sigma}_n)$, each if which is represented by a  spin vector on a reduced Hamming cube, 
   \begin{equation}
   \pmb{\sigma}_k \in \mathcal{Q}^{(k)}_N :=  \{-1,1\}^{\lceil x_k N \rceil - \lceil x_{k-1} N \rceil}, \quad k \in \{ 1, \dots , n \} .
   \end{equation}
Introducing independent standard Gaussian variables $X_{\pmb{\sigma}_1}, X_{\pmb{\sigma}_1\pmb{\sigma}_2}, \ldots, X_{\pmb{\sigma}_1\pmb{\sigma}_2\cdots \pmb{\sigma}_n}$ one then rewrites 
      \begin{equation}\label{eq:ugrem}
   U(\pmb{\sigma}) = \sqrt{N}(\sqrt{a_1} X_{\pmb{\sigma}_1} + \sqrt{a_2}X_{\pmb{\sigma}_1\pmb{\sigma}_2} + \cdots + \sqrt{a_n}X_{\pmb{\sigma}_1\pmb{\sigma}_2\cdots \pmb{\sigma}_n}) 
   \end{equation}
   in the sense of distributional equality.
   The pressure or negative free energy
   \[ \Phi_N(\beta) \coloneqq \frac1N \ln Z_N(\beta)  \] 
   is given in terms of the partition function $  Z_N(\beta) \coloneqq \sum_{\pmb{\sigma} \in \mathcal{Q}_N} e^{-\beta U(\pmb{\sigma})} $, and 
   converges for any distribution function $A$ almost surely \cite{CCP87}.
   The limit depends on the concave hull $\bar{A}$ of $A$, i.e.\ the smallest concave function which is greater or equal than $A$. In the GREM, the concave hull  $\bar{A}$ is a piecewise linear function determined by the values $
   \{y_1,\ldots,y_m  \} \subset \{x_1,\ldots,x_n  \}$ where $A$ and $\bar{A}$ agree. Let us further introduce the following quantities: the increments of the concave hull $\bar{a}_l \coloneqq A(y_l) - A(y_{l-1})$, the interval lengths $L_l \coloneqq y_l - y_{l-1}$ and the slopes $\gamma_l \coloneqq \bar{a}_l/L_l$. The limit of $\Phi_N$ is then given by \cite{DG86,CCP87}
   \begin{equation}
   \lim_{N \to \infty} \Phi_N(\beta) = \Phi(\beta) = \sum_{l=1}^{m} \varphi^{(l)}(\beta)
   \end{equation}
   with the partial pressures
   \begin{equation}\label{eq:partfren}
   \varphi^{(l)}(\beta) \coloneqq \begin{cases}
   \frac12 \beta^2 \bar{a}_l + L_l \ln2 & \text{ if } \beta \leq \sqrt{(2  \ln 2) \gamma_{l}^{-1}} =:\beta_l, \\
   \beta \sqrt{(2 \ln2 )\bar{a}_l L_l} & \text{ if } \beta >  \sqrt{(2  \ln 2) \gamma_{l}^{-1}}.
   \end{cases}
   \end{equation}
    (For future reference, we note that this formula still holds if the weights $ (a_k )$ do not add up to one.) 
    The glass transition in the GREM occurs in steps with the  components of the systems' spins corresponding to $l$ freezing at $ \beta_l $. Since $ \beta_m > \dots > \beta_2 > \beta_1 $, the the highest freezing temperature is found at $ \beta_c = \beta _1 $. \\
     
   The CREM includes distribution functions $A$ which are not step functions. Since they can be represented as a (uniform) limit of step functions, it is not surprising that corresponding limit of the pressure $\Phi(\beta)$ turns into an integral. The increments $\bar{a}_l$ are replaced by the right derivative $\bar{a}(x)$ of $\bar{A}(x)$ which exists everywhere as a consequence of the convexity of $\bar{A}$. This allows one to give an explicit expression for the limit $\Phi(\beta)$ (cf.~\cite{BK04b,Bov06})
   \begin{equation}\label{eq:crem}
   \Phi(\beta) = \sqrt{2 \ln 2} \ \beta \int_{0}^{x(\beta)}
   \sqrt{\bar{a}(x)}\, dx  + \frac{\beta^2}{2}(1-\bar{A}(x(\beta))) + (1- x(\beta)) \ln 2
   \end{equation}
   with the function 
   \begin{equation}\label{eq:xbeta}
   x(\beta) \coloneqq \sup \left\{x \ | \ \bar{a}(x) > (2 \ln2)/ \beta^2 \right\}
   \end{equation}
   The glass transition in the CREM occurs at $ \beta_c = \sqrt{(2\ln2) / \lim_{x\downarrow 0} \bar{a}(x)}  $. 
   
 \subsection{Quantum GREM and CREM  and a Parisi formula}
 If a transversal magnetic field in $x$-direction is turned on, the total Hamiltonian acting on the Hilbert space $\ell^2(\mathcal{Q}_N)$ is
 \begin{equation}
 H_N = U - B,
 \end{equation}
 where $B$ is the sum of the $x$-Pauli matrices $  \pmb{s}_j^{(1)} $ with (possibly random) weights $b_j\in \mathbb{R} $, i.e.
 \begin{equation}\label{def:B} 
  (B \psi)(\pmb{\sigma}) :=    \sum_{j=1}^{N}  b_j \;  \big( \pmb{s}_j^{(1)}  \psi\big)(\pmb{\sigma}) , \qquad  \big( \pmb{s}_j^{(1)}  \psi\big)(\pmb{\sigma}) :=  \psi(\sigma_1,\ldots,-\sigma_j,\ldots,\sigma_N) ,  
 \end{equation}
 and $U$ is some random potential. Before further specifying $U$ and $B$, we want to record a simple observation: the partition function $\tr e^{-\beta(U-B)}$ as well as the diagonal matrix-elements of $ e^{-\beta(U-B)} $ in terms of the standard $ z $-basis $ | \pmb{\sigma} \rangle $ only depends on the absolute values $(|b_j|)$. (Here and in the following we use the bra(c)ket notation for matrix elements.) 
 \begin{lemma}\label{lem:possemgr}
 	Let $U$ be an arbitrary potential on $\mathcal{Q}_N$ and $B,B'$ two transversal field with weights $b_j$ and $b_j'$ which only differ by a sign, i.e. $|b_j| = |b_j^\prime|$ for all $ j$. Then, for all $ \pmb{\sigma} \in  \mathcal{Q}_N$:
 	\begin{equation}\label{eq:possemgr}
 	\langle \pmb{\sigma} | e^{-\beta(U -B)} | \pmb{\sigma} \rangle = 	\langle \pmb{\sigma} |  e^{-\beta(U-B')} | \pmb{\sigma} \rangle . 
 	\end{equation}
 \end{lemma}
 \begin{proof}
 	Expanding the exponential, we write  $\langle \pmb{\sigma} | e^{-\beta(U- B)} | \pmb{\sigma} \rangle$ as a convergent series of terms of the form
 	\begin{equation}\label{eq:expterms} \langle \pmb{\sigma} | A_1 \cdots A_k | \pmb{\sigma} \rangle    \end{equation}
 	where each $A_j$ is either $-U$ or some $b_j \pmb{s}_j^{(1)}$. As $\pmb{s}_j^{(1)}$ flips the sign of the $j$th coordinate $\sigma_j$, the term~\eqref{eq:expterms} vanishes unless each operator $\pmb{s}_j^{(1)}$ occurs $n_j$ times, where $n_j$ is an even number. We conclude that $\langle \pmb{\sigma} | e^{-\beta(U-B)} | \pmb{\sigma} \rangle$ only depends on the squares $b_j^2$ which proves \eqref{eq:possemgr}. 
 \end{proof}
 
If all the weights $b_i \geq 0$ are non-negative, the Trotter product formula shows that $H_N$ generates a positive semigroup, i.e.
 \[ \langle \pmb{\sigma}| e^{-t H_N} | \pmb{\sigma}^\prime \rangle \geq 0    \]
 for any $t \geq 0$ and $\pmb{\sigma},\pmb{\sigma}^\prime \in \mathcal{Q}_N$. This is in general not true for an arbitrary transversal magnetic field $B$, but due to Lemma~\ref{lem:possemgr} we can assume without loss of generality that the weights $ (b_j) $ are indeed non-negative if we are only interested in properties which can be derived from diagonal matrix elements such as the quantum partition function.

 In the remaining part of this section, we restrict ourselves to the case where the weights $(b_j)$ are independent copies of an absolutely integrable random variable $ \mathfrak{b}$ and they shall be independent of the Gaussian potentials $U$. 
 We are mainly interested in the thermodynamic properties of the hierarchical quantum spin glasses which are encoded in the quantum partition function 
 \[ Z_N(\beta, \mathfrak{b}) = \tr \left[ e^{-\beta H_N}\right] \]
 or, equivalently, in the pressure (or negative free energy)  
 \[ \Phi_N(\beta, \mathfrak{b}) = \frac1N \ln  Z_N(\beta, \mathfrak{b}).  \]
In the special case that the weights $ \mathfrak{b} = \Gamma$ are (almost surely) constant, we will sometimes write $B = \Gamma T$ and denote the pressure by $\Phi_N(\beta,\Gamma)$.
  
Our first main result  concerns the free energy of the QGREM, cf.~\eqref{eq:ugrem}. We show that the free energy converges almost surely to a non-random limit, for which we derive 
an explicit expression in terms of the classical partial free energies 
 \eqref{eq:partfren} and the paramagnetic free energy. With the notation of Section~\ref{subsec:GREM}, we have the following:
 \begin{theorem}\label{thm:qgrem}
For the GREM  specified  by $ U $ as in \eqref{eq:ugrem} in terms of its distribution function $ A $, any $\beta \geq 0$ and an absolutely integrable random variable $\mathfrak{b}$ the quantum free energy converges almost surely,
 \begin{equation}\label{eq:qgrem}
 \lim_{N \to \infty} \Phi_N(\beta, \mathfrak{b}) = \Phi(\beta, \mathfrak{b}) \coloneqq \max_{1 \leq K \leq m} \left[  \sum_{l =1}^{K} \varphi^{(l)}(\beta) + (1-y_K) \ee[ \ln\left(  2 \cosh(\beta  \mathfrak{b}) \right)] \right] .
 \end{equation}
 The maximum is taken over all points $   \{y_1,\ldots,y_m  \} $ supporting the convex hull $\bar{A}$ of $ A $. 
 \end{theorem}
 
 The proof of this theorem is found in Section~\ref{sec:proofGREM} below. 
We stress that as in the classical case the concave hull $\bar{A}$ keeps being the determining function for the limit and not $A$ itself. 
The second term in~\eqref{eq:qgrem} is the pressure of the random quantum paramagnet on a Hamming cube given by 
\begin{equation}\label{eq:para}
 p(\beta,\ \mathfrak{b}) := \frac1N \ee\left[ \ln \tr \left[ e^{-\beta  \mathfrak{b}}\right] \right] = \ee\left[ \ln\left(  2 \cosh(\beta  \mathfrak{b}) \right) \right].
\end{equation}
If $ \mathfrak{b} = \Gamma> 0$ is constant, the structure of  the limit in \eqref{eq:qgrem} becomes more transparent if we introduce the critical field strengths 
\[ \Gamma_c^{(l)} \coloneqq \frac1\beta \arcosh\left( \frac12 
\exp\left( \frac{\varphi^{(l)}(\beta)}{L_l}  \right) \right), \quad l \in \{1, \dots, K\}.\]
In this situation, we may rephrase \eqref{eq:qgrem}  as follows:
\begin{corollary}\label{cor:qgrem}
	In the situation of Theorem \ref{thm:qgrem} with $\mathfrak{b} = \Gamma > 0 $:
	\begin{equation}\label{eq:grem2} \Phi(\beta,\Gamma) = \sum_{l=1}^{m} \left( \varphi^{(l)}(\beta) \ \mathbbm{1}_{\Gamma < \Gamma_c^{(l)}} + L_l \ln\left(  2 \cosh(\beta \Gamma)\right) \   \mathbbm{1}_{\Gamma \geq \Gamma_c^{(l)}}\right).  \end{equation}
\end{corollary}
 The proof  is again found in Section~\ref{sec:proofGREM}. 
The free energy coincides with the sum of $m$ weighted and rescaled QREM terms, cf.~\cite{Gold90,MW19}. In particular, there are as many magnetic first-order transitions as second-order glass transitions. The glass transitions continue to occur at the (classical) critical inverse temperatures 
$  \beta_{l}= \sqrt{(2 \ln2)\gamma_l^{-1}} $
as long as $\Gamma < \Gamma_c^{(l)}(\beta_l)$ and disappear for field strengths $\Gamma > \Gamma_c^{(l)}(\beta_l)$. The specific magnetization in $ z $-direction 
\[
m_z(\beta,\Gamma) :=  \frac{1}{\beta} \frac{ \partial }{\partial \Gamma} \Phi(\beta,\Gamma) 
\]
 changes discontinuously at $ \Gamma = \Gamma_c^{(l)} $,  cf.~Figure~\ref{fig:phase1}. The physics described by~\eqref{eq:grem2}  is that of the block or species of spins corresponding to $ l $ flipping into transversal order at $ \Gamma = \Gamma_c^{(l)} $. At temperatures below $ \beta_l^{-1} $, the transition is from spin-glass order to a quantum paramagnet in that block.  It is an interesting question to determine the fate of Parisi's order parameter in this regime. 
 The rigorous classical analysis of this quantity, which partially captures the geometric structure of the Gibbs measure, can be found in~\cite{BK04a}. An extension of this analysis will be the subject of a future work.  \\
 
 \begin{figure}[ht]
 	\begin{center}
 		\includegraphics[width=.55\textwidth] {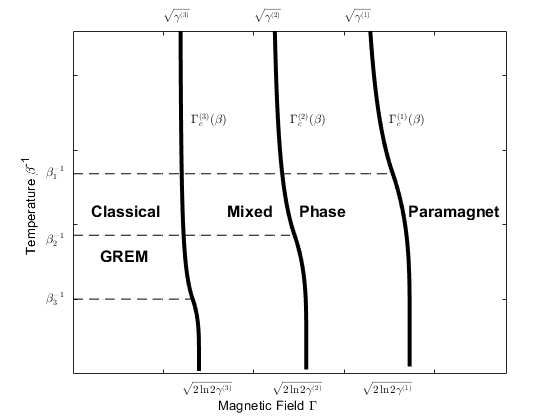}
 		\caption{Phase diagram of the Quantum GREM as a function of the transversal constant magnetic field $ \Gamma  $ and the temperature $ \beta^{-1}$.  The figure shows an example with three second-order glass transitions (dotted lines) and three first-order magnetic transitions
 			(bold lines). If $\Gamma < \Gamma_c^{(3)}(\beta_l)$ the free energy coincides with the classical one ($\Gamma = 0$),
 			whereas for $\Gamma > \Gamma_c^{(l)}(\beta_l)$ the system becomes a pure paramagnet. }\label{fig:phase1}
 	\end{center}
 \end{figure}

Moving on to the more general CREM potentials, it is convenient to introduce truncated versions of the free energy in \eqref{eq:crem}. For any $z \in [0,1]$ we define
\begin{equation}\label{eq:fz}
\Phi(\beta,z) \coloneqq \sqrt{2 \ln 2} \beta \int_{0}^{\min\{x(\beta),z\}}
\mkern-12mu \sqrt{\bar{a}(x)}\, dx  + \mathbbm{1}_{z > x(\beta)}\left( \frac{\beta^2}{2}(\bar{A}(z)-\bar{A}(x(\beta))) + \ln 2(z- x(\beta)\right). \end{equation} As in the quantum GREM, the free energy of the quantum CREM converges almost surely and the limit may be expressed as variational formula involving  $\Phi(\beta,z) $:

\begin{theorem}\label{thm:qcrem}
	For the CREM  specified  by $ U $ as in \eqref{eq:ucrem}  in terms of its distribution function $ A $,  let $\bar{A}$ be the concave hull of $A$, $\bar{a}$ the right-derivative of $\bar{A}$, $\Phi(\beta,z)$ as in \eqref{eq:fz} and $ \mathfrak{b}$ an absolutely integrable random variable. Then, the quantum pressure $\Phi_N(\beta,\ \mathfrak{b})$ converges almost surely,
	\begin{equation}\label{eq:qcrem}
	\lim_{N \to \infty} \Phi_N(\beta,\Gamma) = \Phi(\beta,\Gamma)
	\coloneqq \sup_{0 \leq z \leq 1} \left[ \Phi(\beta,z) + (1-z) \ee \left[ \ln 2 \cosh(\beta  \mathfrak{b}) \right] \right] .
	\end{equation}
\end{theorem}
The proof is found in Section~\ref{sec:proofGREM}.\\

The free energy $\Phi_N(\beta, \mathfrak{b})$ does not only converge almost surely, but also in mean. This is a  consequence of the following Proposition, which is a special case of Proposition \ref{prop:conc2} in Section 3.

\begin{proposition}\label{prop:conc}
For any Gaussian potential $U$ as in \eqref{eq:ucrem}, the Gaussian concentration estimate
\begin{equation}
 \mathbb{P}_U\left(\left| \Phi_N(\beta,\mathfrak{b})  - \mathbb{E}_U\left[\Phi_N(\beta,\ \mathfrak{b}) \right]  \right|  > \frac{t \, \beta }{\sqrt{N} } \right) \leq  2 \,  \exp\left(- \frac{t^2}{4}\right)
 \end{equation}	
 holds for any $t > 0$ and $N \in \nn$. Here, $  \mathbb{P}_U $, $\mathbb{E}_U$ denote the probability and expectation with respect to $U$.
\end{proposition}
If $\mathfrak{b}$ is even an $L^r$-random variable for some $r >1$, we further see that the pressure convergences in $r$-th mean.\\

\begin{figure}
	\begin{center}
		\includegraphics[width=.5\textwidth] {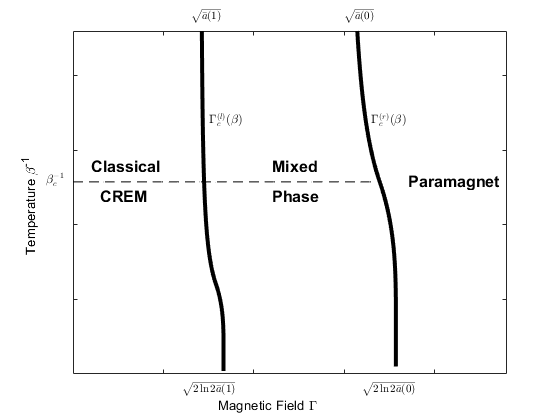}\hfill\includegraphics[width = .5\textwidth]{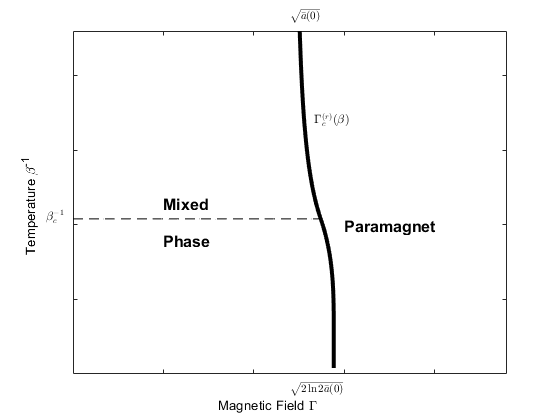}
		\caption{Both figures illustrate examples for the phase diagram of the Quantum CREM as a function of the transversal magnetic field $ \Gamma  $ and the temperature $ \beta^{-1}$. The first plot contains two magnetic phase transitions (bold lines) into transversal magnetic order. The second plot shows the case of one magnetic phase transition. The dotted line corresponds to the glass transition at $ \beta_c = \sqrt{(2\ln2) /\bar{a}(0)} $. If $\bar{A}$ is continuously differentiable, the magnetic transitions are second order.}\label{fig:phase2}
	\end{center}
\end{figure}

In order to determine the order of occurring magnetic phase transitions, we will replace the variational formula \eqref{eq:qcrem} in the case $\mathfrak{b} = \Gamma$ by a more explicit expression. To this end, we assume from now that concave hull $\bar{A}$ is a continuously differentiable function being different from the identity (in order to exclude the QREM situation). Then, $ \Phi(\beta,z)$ is differentiable with respect  to $z$,
\[ \frac{\partial \Phi(\beta,z) }{\partial z} = \sqrt{(2\ln2) \bar{a}(z)} \ \beta \ \mathbbm{1}_{z < x(\beta)} + 
\left( \ln 2 + \frac{\beta^2}{2}  \bar{a}(z) \right) \mathbbm{1}_{z \geq x(\beta)}.  \]

We note that the derivative $\frac{\partial \Phi(\beta,\cdot) }{\partial z} : [0,1] \to [s(\beta), t(\beta)]$ is a nondecreasing continuous function, where we have introduced the boundary values
\[ s(\beta) \coloneqq \frac{\partial \Phi(\beta,z) }{\partial z}\big|_{z = 1} \text{ and } t(\beta) \coloneqq \frac{\partial \Phi(\beta,z) }{\partial z}\big|_{z = 0}.  \]
\begin{corollary}\label{cor:qcrem}
Let $g(\beta,\cdot): [s(\beta),t(\beta)] \to  [0,1]  $ be a (generalized) inverse of the derivative $\frac{\partial \Phi(\beta,z) }{\partial z}$ as a function of $z$. Then,
\begin{equation}\label{eq:qcrem2}
\Phi(\beta,\Gamma) = \begin{cases}
\Phi(\beta,1) & p(\beta \Gamma) \leq s(\beta), \\
\Phi(\beta, g_{\beta}(p(\beta \Gamma))) + (1-g_{\beta}(p(\beta \Gamma))) p(\beta \Gamma) & s(\beta) < p(\beta \Gamma) < t(\beta), \\
p(\beta \Gamma) & t(\beta) \leq p(\beta \Gamma)
\end{cases}
\end{equation}
with the paramagnetic pressure $p(\beta \Gamma)= \ln 2 \cosh(\beta \Gamma)$.
\end{corollary}
Corollary \ref{cor:qcrem} implies that there are either one or two magnetic phase transitions, depending on $s(\beta)$. If $s(\beta) = \ln 2$ or, equivalently, $\bar{a}(1) = 0$, we find a single magnetic phase transitions at the critical magnetization 
\[ \Gamma_c^{(r)}(\beta) = \frac{1}{\beta} 
\arcosh\left(\frac12 e^{t(\beta)}  \right). \]
Otherwise, there is a second phase transition at
\[ \Gamma_c^{(l)}(\beta) = \frac{1}{\beta} 
\arcosh\left(\frac12 e^{s(\beta)}  \right). \]

An explicit computation using \eqref{eq:qcrem2} shows that the specific magnetisation in the transversal direction 
\[
m_z(\beta, \Gamma) = \frac{1}{\beta} \frac{ \partial }{\partial \Gamma} \Phi(\beta,\Gamma)  = \begin{cases}
 0  & p(\beta \Gamma) \leq s(\beta), \\
 (1-g_{\beta}(p(\beta \Gamma))) \tanh(\beta \Gamma) & s(\beta) < p(\beta \Gamma) < t(\beta), \\
\tanh(\beta \Gamma) & t(\beta) \leq p(\beta \Gamma)
\end{cases}
\]
is continuous. 
This transversal magnetic order does not vanish over the line $ \Gamma_c^{(r)}(\beta)  $ but rather only at $  \Gamma_c^{(l)}(\beta)  $ (which is absent in case $\bar{a}(1) = 0$). 
If the derivative of $\bar{a}(x)$ exists at $x=0$ or $x=1$, the second derivative of $\Phi(\beta,\Gamma)$ has a jump at the respective critical magnetic fields and we have a second-order magnetic transition and not first order as in the Quantum GREM, cf. Figure~\ref{fig:phase2}. 
In the classical model, the low-temperature glass phase is described by a random probability measure which captures the distribution of the spin overlaps \cite{BK04b,Bov06}. As with the GREM, it is an interesting question, which will be postponed to a future work, to study the influence of the transversal field on these quantities

\subsection{Quantum Parisi formula for the non-hierarchical GREM }

The non-hierarchical GREM was introduced in \cite{BK06} to illustrate Parisi's ultrametricity conjecture in an explicitly solvable model. We, similarly, want to study the non-hierarchical GREM with a transverse field, as we can explicitly determine the free energy. This is a basic test whether our results in Section~\ref{sec:peeling} are only strictly valid for hierarchical models or if one might hope that they still hold true to a certain extent for more complicated models.

As in the GREM we write $\pmb{\sigma} = \pmb{\sigma}_1 \ldots \pmb{\sigma}_n$ 
with $\pmb{\sigma}_k \in \mathcal{Q}^{(k)}_N$ and $L_{k} = x_{k} - x_{k-1}$ are interval lengths. We denote by $\mathcal{P}$ the power set of $\{1,\ldots,n\}$. To each subset $J = \{j_1,\ldots,j_m\} \in \mathcal{P}$ we assign the spin vector 
\[ \pmb{\sigma}_{J} = \pmb{\sigma}_{j_1}\ldots\pmb{\sigma}_{j_m}  \] 
and a nonnegative number $a_J \geq 0$ with $ a_\emptyset = 0 $. We further assume that the numbers $a_J$ add up to one, $ \sum_{J\in \mathcal{P}} a_J = 1 $. 
For each $J \in \mathcal{P}$ we denote by $X_{\pmb{\sigma}_J}^J$ independent standard Gaussian variables, which are also independent from each other
for different indices $J_1 \neq J_2$. The total Gaussian process $V$ on $\mathcal{Q}_N$ is given by
\[ V(\pmb{\sigma}) = \sqrt{N} \sum_{J \in \mathcal{P}} \sqrt{a}_{J} X_{\pmb{\sigma}_J}^J.  \]
The GREM from the previous sections is the special case where $a_{J} \neq 0$ only if $J = \{1,\ldots,k \}$ for $k \in \{ 1, \dots,  n\} $.
The reason for introducing this class of spin-glass models is that in contrast to the GREM the induced metric $\mathbb{E}\left[ |V(\pmb{\sigma}) - V(\pmb{\sigma^\prime}) |^2\right]^{1/2} $ 
does in general not satisfy ultrametricity, cf.~\cite{BK06}.

 In order to recall from~\cite{BK06} the  explicit expression for the limit of the classical free energy, we introduce the notion of a chain $\mathcal{C}$. A chain 
$S = \{A_0, A_1, A_2,\ldots, A_n \} \subset \mathcal{P}$
consists of nested sets $A_i$, i.e.
\[ \emptyset = A_0 \subset A_1 \subset A_2 \subset \ldots \subset A_n, \]
with cardinality $|A_i| = i$. To each chain 
$S = \{A_1,A_2,\ldots,A_n \} $ we assign a hierarchical GREM with weights
\[ a_k^{S} \coloneqq \sum_{A_{k-1} \subset D \subset A_k} a_D \]
and endpoints
\[ y_k^S \coloneqq \sum_{j \in A_k} L_j. \]
Note that $ \sum_{k = 1}^n  a_k^{S} = 1 $ for any chain $ S $ by construction. 
The corresponding hierarchical GREM's pressure converges and we denote the limit by $\Phi(\beta,S)$. Bolthausen and Kistler showed in \cite{BK06} that the limit of  the pressure in the non-hierarchical GREM converges as well. More precisely, we have
\begin{equation}
\Phi(\beta) \coloneqq \lim_{N \to \infty} \frac1N \ln \tr e^{-\beta V} = \min_{S \in \mathcal{C}} \Phi(\beta,S),
\end{equation}
where the minimum is taken over all chains. 

After these preparations, we are able to consider the non-hierarchical GREM with a transverse magnetic field whose weights are independent copies of some random variable $\mathfrak{b}$. We define the pressure as before,
\[ \Phi_N(\beta,\mathfrak{b}) \coloneqq \frac1N \ln \tr e^{-\beta(V-B)}. \]
The general theory developed  in Section~\ref{sec:peeling} also applies to the non-hierarchical GREM and yields

\begin{theorem}\label{thm:nonhgrem}
	Let $\beta \geq 0$ and $\mathfrak{b}$ an absolutely integrable random variable. Then, the pressure $ \Phi_N(\beta,\mathfrak{b})$ converges almost surely and the limit is given by 
	\begin{equation}\label{eq:nonhgrem}
	\Phi(\beta,\mathfrak{b}) \coloneqq \lim_{N \to \infty} \Phi_N(\beta,\mathfrak{b}) = \max_{D \in \mathcal{P}} \min_{S \in \mathcal{C}^{D}} \left[ \Phi_D(\beta,S)   + \sum_{k \in D^c} L_k \ \ee[\ln 2\cosh(\beta \mathfrak{b})] \right].
	\end{equation}
	Here, $\mathcal{C}^{D}$ denotes the set of chains which end at $D$, i.e.\ $S = \{A_0, A_1,\ldots A_m\} \in \mathcal{C}^{D}$ if and only if 
	\[ \emptyset = A_0 \subset A_1 \subset \cdots \subset A_m = D \]
	and $|A_i| = i$ for any $1 \leq i \leq m$. Moreover, $\Phi_D(\beta,S)$ is the pressure of the corresponding GREM on the reduced hypercube associated to $ D $. 
\end{theorem} 

The proof of this theorem is found in Section~\ref{sec:nonhier}.

The max-min structure of the limit in \eqref{eq:nonhgrem} seems to be quite universal as it also appears in the Parisi's formula for vector spin glasses \cite{Pan14}. This formula was used in \cite{AB19} to obtain an expression for the limit in the quantum SK-model. However, the maximum is essentially taken over the infinite-dimensional path overlap, which makes it hard to analyze. One might hope to find a less involved parametrization of the overlap distribution which is easier to access.

In fact, \eqref{eq:nonhgrem} can be further simplified: the limit $	\Phi(\beta,\mathfrak{b})$ does only depend on a single chain $S$.

\begin{corollary}\label{cor:nonhgrem}
	There exists a chain $S\in \mathcal{C}$ such that for any $\beta \geq 0$ and absolutely integrable variable $\mathfrak{b}$, 
	\begin{equation}\label{eq:cornonh}
      \Phi(\beta,\mathfrak{b}) =  \Phi(\beta,\mathfrak{b},S).
	\end{equation}
	Here $\Phi(\beta,\mathfrak{b},S)$ denotes the pressure of Quantum GREM assigned to $ S $, cf.~\eqref{eq:qgrem}.
\end{corollary}

Corollary~\ref{cor:nonhgrem}, whose proof is also found in Section~\ref{sec:nonhier}, shows that the non-hierarchical GREM in a transversal field is at least on a thermodynamical level equivalent to an ordinary Quantum GREM.

 \section{Hierarchical peeling}\label{sec:peeling}
 In this section, we present the general principle, which we dubbed hierarchical peeling, from which the main results presented in the previous section will follow.
 We first describe the core of this idea in the binary setup.
 
 \subsection{Peeling principle}
 
We start by describing the general setting. Picking a parameter  $0 < x \leq 1$, we will decompose the hypercube $  \mathcal{Q}_N $  into two reduced hypercubes on arrays of length $ \lceil x N \rceil $ and $ N - \lceil x N \rceil $. 
Accordingly, we write $\pmb{\sigma} = ( \pmb{\sigma}_1 , \pmb{\sigma}_2) \in \mathcal{Q}_N $, where $\pmb{\sigma}_1 \in \mathcal{Q}_N^{(1)} := \mathcal{Q}_{\lceil xN  \rceil  }$ and $\pmb{\sigma}_2 \in  \mathcal{Q}_N^{(2)}  := \mathcal{Q}_{N - \lceil x N \rceil }$. We consider Hamiltonians $H = U -B$, where $U$ is a random potential on $\mathcal{Q}_N$ and $B$ is a random transversal field. We will need to require several assumptions on $H$. We start with $U$:

\begin{assumption}[Assumptions on $U$] \label{ass1}
	The random potential $U(\pmb{\sigma})$ on $\mathcal{Q}_N$ takes the form
	\begin{equation}
	U(\pmb{\sigma}) = V_N(\pmb{\sigma})  + X_{\pmb{\sigma}_1} 
	\end{equation}
	with some random potential $V_N$ which is independent of the random process $X_{\pmb{\sigma}_1}$. The random variables $X_{\pmb{\sigma}_1}$ with $\pmb{\sigma}_1 \in \mathcal{Q}_N^{(1)}$ are absolutely integrable, centered, and satisfy: 
	\begin{itemize}
		\item[1.]  $X_{\pmb{\sigma}_1}$ are independent and identically distributed for each fixed $N \in \nn$. 
		\item[2.] The pushforward measures $\mu_N$ of the negative parts $X^{-}_{\pmb{\sigma}_1}/N$ satisfy a large deviation principle (LDP) with a lower semi-continuous rate function $I \colon \rr \to [0,\infty]$, i.e., for any Borel set $\mathcal{B} \subset \rr$,
		\begin{equation}\label{eq:largedev1}
		- \inf_{x \in \mathrm{int}(\mathcal{B})} I(x) \leq \liminf_{N \to \infty} \frac1N \ln \mu_N(\mathcal{B}) \leq \limsup_{N \to \infty} \frac1N \ln \mu_N(\mathcal{B}) \leq - \inf_{x \in \mathrm{clos}(\mathcal{B})}
		I(x). \end{equation}
		Moreover,  we assume 
		\begin{equation}\label{eq:largedev2}
		\inf_{ x \in (-\infty,-\epsilon]} I(x) > 0
		\end{equation}
		for any $\epsilon> 0$.
		\item[3.]  For any random weights 
		$w_{\pmb{\sigma}_1} $ which are independent from $X_{\pmb{\sigma}_1}$
		and further fulfill almost surely 
		\[ w_{\pmb{\sigma}_1} \geq 0, \quad \sum_{\pmb{\sigma}_1 \in \mathcal{Q}_N^{(1)}} w_{\pmb{\sigma}_1} = 1,  \]
		a generalized strong law holds true almost surely, that is, 
		\begin{equation}\label{eq:gslln}
		\lim_{N \to \infty} \frac1N \sum_{\pmb{\sigma}_1 \in \mathcal{Q}_N^{(1)}} w_{\pmb{\sigma}_1  } X_{\pmb{\sigma}_1} = 0.
		\end{equation}
	\end{itemize}	
\end{assumption}
As will be discussed in the next subsection, we will mostly be interested in hierarchical $V_N$ as in the GREM or CREM, but our results also applies to the more general situation. An important example where $V_N$ is not of CREM type is the case of a non-vanishing longitudinal magnetic field. \\
The LDP \eqref{eq:largedev1} with \eqref{eq:largedev2} ensure that  probabilities of the type $\pp(X_{\pmb{\sigma}_1} < - \epsilon N)$ decay exponentially in $N$ for any $\epsilon > 0$. Assumption \eqref{eq:gslln} is a technical condition needed for our proof of Theorem~\ref{prop:key}. 
The following examples of random variables $X_{\pmb{\sigma}_1}$ meet Assumption~\ref{ass1} which can be seen by the sufficient criterion Lemma~\ref{lem:slln} that we present in the appendix :
\begin{itemize}
	\item[1.] $X_{\pmb{\sigma}_1} = \sqrt{Na} Y_{\pmb{\sigma}_1}$ with independent standard Gaussian $Y_{\pmb{\sigma}_1}$ and some $a > 0$. The rate function of the negative part is $I(x) = \frac{a}{2} x^2 \mathbbm{1}_{x < 0}$.
	\item[2.] Another example is  $X_{\pmb{\sigma}_1} = -N Y_{\pmb{\sigma}_1}$, where $Y_{\pmb{\sigma}_1}$ are independent and follow an exponential distribution with parameter $N$. The rate function of the negative part is $I(x) = |x| \mathbbm{1}_{x < 0}$.
	\item[3.] More generally, let $Y \leq 0$ be a random variable with a decay of the form $\exp(-C t^{\alpha})$ for some $\alpha, C > 0$, that is, 
	\[ - \lim_{t \to  \infty} t^{-\alpha} \ln \pp(Y <  t) = C. \]
	Then, we define $X_{\pmb{\sigma}_1} = N^{1-1/\alpha} Y_{\pmb{\sigma}_1}$, where $Y_{\pmb{\sigma}_1}$ are independent copies of $Y$. The corresponding rate function is given by $I(x) = C |x|^{\alpha} \mathbbm{1}_{x < 0}$. 
\end{itemize}

We consider a not necessarily constant transversal magnetic field $ B = \sum_{j=1}^{N} b_j \pmb{s}_j^{(1)} $ as in \eqref{def:B}  with $(b_j)$  random variables which do not need to be independent from each other. 
The transversal field $ B $  splits  into two parts $B^{1,x} $ and $B^{2,x}$, which act exclusively on the respective part of the array, 
 \[ \begin{split}
 B^{1,x} |\pmb{\sigma} \rangle &\coloneqq \sum_{i=1}^{\lceil x N \rceil} b_i \pmb{s}_i^{(1)} |\pmb{\sigma} \rangle =  \big|  B_{\big| \mathcal{Q}_N^{(1)} } \pmb{\sigma}_1, \pmb{\sigma}_2 \rangle, \\
 B^{2,x} |\pmb{\sigma} \rangle &\coloneqq  \sum_{i=\lceil x N \rceil+1}^{N} b_i \pmb{s}_i^{(1)} |\pmb{\sigma} \rangle  =   \big|\pmb{\sigma}_1, B_{\big| \mathcal{Q}_N^{(2)} }\pmb{\sigma}_2 \rangle.
 \end{split}   \]
 If $x =1$, we simply set $B^{2,x} = 0$.
Subsequently, we assume the following on the transversal field $B$:

\begin{assumption}[Assumptions on $B$]\label{ass2}
  The random weights $(b_j)$ are independent of the potential $U$ and  satisfy almost surely 
  \begin{equation}\label{eq:ass2}
  \limsup_{N \to \infty} N^{-1} \sqrt{\sum_{j=1}^{N} |b_j|^2} = 0.
  \end{equation}
\end{assumption}

Let us discuss some sufficient conditions on the weights $(b_j)$ which ensure the validity of Assumption~\ref{ass2}:

\begin{itemize}
	\item[1.] Assumption~\ref{ass2} obviously covers the constant field case $b_j = \Gamma \geq 0$.
	\item[2.] If the weights are almost surely dominated by $\sqrt{N}$, that is,
	\begin{equation}
	\limsup_{N \to \infty} N^{-\frac12} \max_{1\leq j \leq N} |b_j| = 0,
	\end{equation}
	then \eqref{eq:ass2} holds true.
	\item[3.] In view of the framework in Section~\ref{sec:intro}, we are mostly interested in weights $(b_j)$ forming independent copies of an absolute integrable random variable $\mathfrak{b}$. Then, \eqref{eq:ass2} is satisfied and this result is presented as Lemma~\ref{lem:upb} in the appendix. If we additionally assume that $\ee[|\mathfrak{b}|^r]$ is finite for some $r >1$, Assumption~\ref{ass2} is easily verified. If $r \in [1,2]$, then
	\[ N^{-1} \sqrt{\sum_{i=1}^{N} |b_i|^2} \leq N^{-(1-1/r)} \left(N^{-1}\sum_{i=1}^N |b_i|^r\right)^{1/r}.  \]
	The term in the bracket converges almost surely to a constant by the strong law of large numbers. So \eqref{eq:ass2} is fulfilled .
\end{itemize}

If Assumption \ref{ass1} and \ref{ass2} holds true, our main results states that  the pressure 
\[ \Phi_N(\beta) \coloneqq \frac1N \ln  \tr\left[ e^{-\beta (U - B ) } \right] \]
asymptotically agrees with the maximum of the partial pressures
\[ \Phi_N^{(1)}(\beta) \coloneqq \frac1N \ln \tr\left[e^{-\beta(V_N - B)}   \right] \quad
\mbox{and} 
\quad \Phi_N^{(2)}(\beta)  \coloneqq  \frac1N \ln  \tr\left[ e^{-\beta (U - B^{2,x} ) } \right] \]
even if $\Phi_N^{(1)}(\beta)$ and $\Phi_N^{(2)}(\beta)$ do not converge:
 \begin{theorem}\label{prop:key}
 	Under Assumption~\ref{ass1} and \ref{ass2} for any $ x \in (0,1] $ we have the almost sure convergence
 	\begin{equation}\label{eq:prop}
 	\limsup_{N \to \infty} |\Phi_N(\beta) -  \max\{\Phi_N^{(1)}(\beta), \Phi_N^{(2)}(\beta)  \}| = 0.
 	\end{equation}
 	 
 \end{theorem}
Roughly speaking, the Gaussian variables $(X_{\pmb{\sigma}_1})$ and the partial magnetic term $B^{1,x}$ only contribute separately from each other to the free energy.
This result may be regarded as a generalization of the limit theorem for the QREM in \cite{MW19}. Let us remark that if 
	the almost-sure limits
	\[ \Phi^{(1)}(\beta) := \lim_{N \to \infty} \Phi_N^{(1)}(\beta), \quad \mbox{and}\quad  
	\Phi^{(2)}(\beta) := \lim_{N \to \infty} \Phi_N^{(2)}(\beta)  \]
	exist for any $\beta \geq 0$, we immediately obtain 
	\begin{equation}\label{eq:prop2}
	\lim_{N \to \infty} \Phi_N(\beta) = \max\{\Phi^{(1)}(\beta), \Phi^{(2)}(\beta)  \}.
	\end{equation}

For a proof of~Theorem~\ref{prop:key} the methods in \cite{MW19} are robust enough to be extended.
We briefly recall some notations and results necessary. For $ \epsilon > 0 $ we denote  the large deviation set of $  X_{\pmb{\sigma}_1} $ by
\begin{equation}
 \mathcal{L}_\epsilon \coloneqq \left\{\pmb{\sigma}_1 \in \mathcal{Q}_{N}^{(1)} \big| \  X_{\pmb{\sigma}_1} \leq - \epsilon N \right\}  . 
 \end{equation}
The difference between $B^{1,x}$ and its Dirichlet restrictions  to this large deviation set is
\[   A_{\mathcal{L}_\epsilon} := B^{1,x} - B^{1,x}_{ \mathcal{L}_\epsilon^c} , \] 
which acts non-trivially only on the first component  $\mathcal{H}^1 :=  \ell^2(\mathcal{Q}_{N}^{(1)} ) $ of the tensor-product Hilbert space 
$ \ell^2(\mathcal{Q}_N) = \mathcal{H}^1  \otimes \mathcal{H}^2 $ with $  \mathcal{H}^2 :=  \ell^2( \mathcal{Q}_{N}^{(1)} ) $.
We will need the following generalization of Lemma 2 and Lemma~3 in \cite{MW19}:
\begin{proposition}\label{prop:oldpap}
Under Assumption~\ref{ass1} and \ref{ass2}, for any $ \epsilon > 0 $ and $ x \in (0, 1] $ we have almost surely
\begin{equation} \label{eq:norm}  \limsup_{N \to \infty} N^{-1} \|A_{\mathcal{L}_\epsilon}\|  = 0.  \end{equation}
\end{proposition}

 The proof of Proposition~\ref{prop:oldpap} is based on an estimate for the maximal size of the so-called gap-connected components of $\mathcal{L}_\epsilon$, which are defined as follows:

\begin{definition}
	Let $\widetilde{\mathcal{Q}}^{(1)}_{N}$ be the supergraph of the Hamming cube $\mathcal{Q}_{N}^{(1)}$, which one obtains by adding the edges $\{ \pmb{\sigma}_1, \pmb{\sigma}_1'  \}$, where $\pmb{\sigma}_1,\pmb{\sigma}_1'$ are two vertices with $d(\pmb{\sigma}_1,\pmb{\sigma}_1') =2$. 
	We call $   \mathcal{C}_\varepsilon \subset \mathcal{L}_\varepsilon $  a gap-connected component, if $   \mathcal{C}_\varepsilon$  is connected as a subset of $\widetilde{\mathcal{Q}}_N^{(1)}$.
	A gap-connected component $  \mathcal{C}_\varepsilon $ is maximal if there is no other vertex $ \pmb{\sigma}_1\in  \mathcal{L}_\varepsilon \backslash  \mathcal{C}_\varepsilon $ such that $  \mathcal{C}_\varepsilon \cup \{ \pmb{\sigma}_1 \} $ 
	forms a gap-connected component. We denote by $\left( \mathcal{C}_\epsilon^{\alpha}\right)_\alpha $ the maximal gap-connected components of $\mathcal{L}_\varepsilon$.
\end{definition}

We claim that the maximum of the cardinality $\max_{\alpha} \left| \mathcal{C}_\epsilon^{\alpha} \right| $ is almost surely of order one:
\begin{lemma}\label{lem:gapc}
	Under Assumption~\ref{ass1}, 
	 for any $ \epsilon > 0 $ and $ x \in (0, 1] $ there is $K > 0$ such that 
	\begin{equation} \label{eq:gap}  
	\limsup_{N \to \infty} \max_{\alpha} \left| \mathcal{C}_\epsilon^{\alpha} \right| \leq K
	 \end{equation}
	holds almost surely.
\end{lemma}

\begin{proof}
	We fix $K \in \nn$ and introduce the event 
	\[ \Omega_{\epsilon,K,N} \coloneqq \bigcap_{\pmb{\sigma}_1 \in \mathcal{Q}_{N}^{(1)}} \{|B_{4K}(\pmb{\sigma}_1) \cap \mathcal{L}_{\epsilon}| < K \}.  \]
	We note that for $\omega \in \Omega_{\epsilon,K,N}$ we always 
	have $\max_{\alpha} |C_\epsilon^{\alpha}| < K$ as any gap-connected component with $K$ vertices is contained in some ball of radius $4K$. We estimate the probability of the complement $\Omega_{\epsilon,K,N}^{c}$ using the union bound:
	\[ \begin{split} 
	\pp(\Omega_{\epsilon,K,N}^{c}) \leq \sum_{\pmb{\sigma}_1 \in \mathcal{Q}_{N}^{(1)}} \pp(|B_{4K}(\pmb{\sigma}_1) \cap \mathcal{L}_{\epsilon}| \geq K) \leq 2^{\lceil xN \rceil }  \binom{| B_{4K} | }{K}
	\pp(X_{\pmb{\sigma}_1} < - \epsilon N)^K
	\end{split}    \]
	The second inequality follows from the independence of the random variables $X_{\pmb{\sigma}_1}$.  
	The rate function $I$ of $X_{\pmb{\sigma}_1}/N$ satisfies $ \inf_{- \infty < z \leq - \epsilon} I(z) = \delta_{\epsilon} > 0 $, from which we conclude
	\[ \begin{split}  
	\pp(\Omega_{\epsilon,K,N}^{c}) \leq 2^{\lceil xN \rceil }  \frac{| B_{4K} |^K }{K!} e^{-KN(\delta_\epsilon + o(1) )} \exp\left( | B_{4K} | e^{- N (\delta_\epsilon + o(1) )} \right) .
		\end{split}  \]
	Since $ | B_{4K} | \leq e N^{4K} $, we may choose $K = K(\epsilon)$ large enough such that this probability decays exponentially fast. A Borel-Cantelli argument then yields the almost-sure bound
	\[ \limsup_{N \to \infty} \max_{\alpha} |C_\epsilon^{\alpha}| \leq K.    \]
\end{proof}

Proposition~\ref{prop:oldpap} is now a simple consequence of Assumption~\ref{ass2} and Lemma~\ref{lem:gapc}:
\begin{proof}[Proof of Proposition \ref{prop:oldpap}]
	
The operator $A_{\mathcal{L}_\epsilon}$ exhibits a natural decomposition as direct sum 
\[ A_{\mathcal{L}_\epsilon} = \bigoplus_{\alpha} A_{\mathcal{C}_\epsilon^{\alpha}},  \]	
where $A_{\mathcal{C}_\epsilon^{\alpha}}$ denotes the restriction of $A_{\mathcal{L}_\epsilon}$ to the maximal gap-connected component $\mathcal{C}_\epsilon^{\alpha}$.
The Frobenius norm bound 
\[ \|A_{\mathcal{L}_\epsilon} \| \leq \max_{\alpha} \|A_{\mathcal{C}_\epsilon^{\alpha}}\| \leq \sqrt{2 \max_{\alpha} |C_\epsilon^{\alpha}| \sum_{i=1}^{N} |b_i|^2},  \]	
together with Assumption~\ref{ass2} and Lemma~\ref{lem:gapc} completes the proof.
\end{proof} 

 We finally spell out the proof of Theorem~\ref{prop:key}:
  \begin{proof}[Proof of Theorem~\ref{prop:key}]
 	We separately establish an asymptotically sharp lower and upper bound.
 	
 	\noindent
 	\textit{Lower bound:} The lower bound rests on a twofold application of Gibbs' variational principle.
 	Let first $ \rho_\beta^{(1)}$ be the Gibbs state of the Hamiltonian $H^{(1)} = V_N - B$. Then, an 
 	application of the Gibbs variational principle (with $\rho = \rho_\beta^{(1)} $ and $ H = H^{(1)} + X_{\pmb{\sigma}_1} $) yields
 	\[ \Phi_N(\beta)  = N^{-1} \sup_\rho  \left[\beta \, \tr \left(H \rho \right) +  \tr\left( \rho \ln \rho \right) \right]   \geq \Phi_N^{(1)}(\beta) +  N^{-1}
 	\sum_{\pmb{\sigma}_1} X_{\pmb{\sigma}_1} w_{\pmb{\sigma}_1} .  \]
 	The weights $w_{\pmb{\sigma}_1} \coloneqq \sum_{\pmb{\sigma}_2} \langle \pmb{\sigma}_1 \pmb{\sigma}_2| \rho^{(1)}_{\beta} | \pmb{\sigma}_1 \pmb{\sigma}_2 \rangle $ are nonnegative, add up to one, and are independent of $X_{\pmb{\sigma}_1}$. By Assumption~\ref{ass1} we conclude that almost surely
 	\[ \liminf_{N \to \infty} \left( \Phi_N(\beta,\Gamma) -  \Phi_N^{(1)}(\beta) \right) \geq 0.  \]
 	
 	Next, the eigenstates $|\psi \rangle \in \ell^2(\mathcal{Q}_N)$ of $H^{(2)} = U - B^{2,x} $ take the form of tensor products
 	$ |\psi \rangle = |\pmb{\sigma}_1 \rangle \otimes |\phi \rangle   $
 	with a certain $|\phi \rangle \in  \ell^2(\mathcal{Q}^{(2)}_{N} )$.  As the matrix elements $\langle \psi | B^{1,x} |\psi \rangle$ vanish for these eigenstate $|\psi \rangle$, 
 	the Gibbs state $ \rho_\beta^{(2)} = e^{-\beta H^{(2)}} / \tr  e^{-\beta H^{(2)}}  $ satisfies
 	\begin{equation}\label{eq:Gibbs1} \tr  B^{1,x} \rho_\beta^{(2)} = 0 .  
	\end{equation}
 	The Gibbs variational principle (for $\rho = \rho_\beta^{(2)} $ and $ H = H^{(2)} - B^{1,x} $) again yields
 	\begin{equation}\label{eq:Gibbs2} \Phi_N(\beta)  \geq \Phi_N^{(2)}(\beta) .  \end{equation}
 	Combining both lower bounds, we obtain almost surely
 	\[ \liminf_{N \to \infty} \left( \Phi_N(\beta,\Gamma) -  \max\{\Phi_N^{(1)}(\beta), \Phi_N^{(2)}(\beta) \} \right) \geq 0  \]
 	
 	\noindent
 	\textit{Upper bound:}  
 	Let $\epsilon > 0$ be arbitrary and consider the direct sum decomposition of the Hilbert space 
 	$ \ell^2(\mathcal{Q}_N) =  \left( \ell^2(\mathcal{L}_\epsilon)  \otimes \mathcal{H}^2 \right)  \oplus \left(\ell^2(\mathcal{L}_\epsilon^c)  \otimes \mathcal{H}^2 \right) $. 
 	The only term in $ H $ connecting the two subspaces is $ A_{\mathcal{L}_\epsilon} $. The Golden-Thompson inequality together with trivial norm estimates  thus yields
 	\[ \tr e^{-\beta(U - B)} \leq e^{\beta  \| A_{\mathcal{L}_\epsilon}\|} \left( \tr_{\big| \ell^2(\mathcal{L}_\epsilon) \otimes  \mathcal{H}^2} e^{-\beta(U -  B^{2,x})} +  e^{\beta \epsilon N} \tr_{\big| \ell^2(\mathcal{L}_\epsilon^c) \otimes   \mathcal{H}^2} e^{-\beta(V_N -  B_{\mathcal{L}_\epsilon^c}^{1,x} - B^{2,x})}   \right) .  \]
 	In the last term we additionally used the fact that $  X_{\pmb{\sigma}_1}  \geq -\epsilon N $ for all $ \pmb{\sigma}_1 \in \mathcal{L}_\epsilon $.
 	The first term is bounded by
 	\[ \tr_{\big|  \ell^2(\mathcal{L}_\epsilon) \otimes  \mathcal{H}^2} e^{-\beta(U - B^{2,x})} \leq \tr e^{-\beta(U- B^{2,x})}.  \]
 	The  second term  is estimated using the non-negativity of the diagonal matrix elements of the semigroups generated by $B$ and the Golden-Thompson inequality again 
 	\[ \tr_{\big| \ell^2(\mathcal{L}_\epsilon) \otimes  \mathcal{H}^2} e^{-\beta(V_N - B_{\mathcal{L}_\epsilon^c}^{1,x} - B^{2,x})} \leq \tr e^{-\beta(V_N - B_{\mathcal{L}_\epsilon^c}^{1,x} - B^{2,x})}  \leq e^{\beta \|A_{\mathcal{L}_\epsilon\|} }\  \tr e^{-\beta(V_N - B)}  \]
 	Since $\epsilon > 0$ was arbitrary and $\|A_{\mathcal{L}_\epsilon}\| = o(N)$ by Proposition~\ref{prop:oldpap}, we conclude the almost-sure inequality 
 	\[ \limsup_{N \to \infty} \left( \Phi_N(\beta)-  \max\{\Phi_N^{(1)}(\beta), \Phi_N^{(2)}(\beta) \} \right)   \leq 0.  \]
 \end{proof}

\subsection{Application to QGREM and QCREM} 
Since we are free in the choice of $V_N$ in Proposition \ref{prop:key}, we obtain the following corollary for GREM type potentials:
\begin{corollary}\label{cor:key}
Let $X = \sqrt{a_1} X_{\pmb{\sigma}_1} + \sqrt{a_2}X_{\pmb{\sigma}_1\pmb{\sigma}_2} + \cdots + \sqrt{a_n}X_{\pmb{\sigma}_1\pmb{\sigma}_2\cdots \pmb{\sigma}_n}$ be a Gaussian vector as in \eqref{eq:ugrem} and $V_N$ an independent potential.  Then, we have the almost sure convergence
\begin{equation}\label{eq:keycor}
\limsup_{N \to \infty} \left| \frac1N \ln \tr e^{-\beta(\sqrt{N} X +V_N -B)} - \max_{0 \leq k \leq n} \frac1N \ln \tr e^{-\beta(\sqrt{N} (\sqrt{a_1} X_{\pmb{\sigma}_1} \cdots + \sqrt{a_k}X_{\pmb{\sigma}_1\pmb{\sigma}_2\cdots \pmb{\sigma}_k}) + V_N - B^{2,x_k})} \right| = 0.
\end{equation}	
\end{corollary}
\begin{proof}
We first apply Theorem~\ref{prop:key} to $\sqrt{a_n}X_{\pmb{\sigma}_1\pmb{\sigma}_2\cdots \pmb{\sigma}_n}$ and $V_N^{(n)} := V_N + \sqrt{Na_1 } X_{\pmb{\sigma}_1} + \sqrt{Na_2 }X_{\pmb{\sigma}_1\pmb{\sigma}_2} + \cdots + \sqrt{Na_{n-1} }X_{\pmb{\sigma}_1\pmb{\sigma}_2\cdots \pmb{\sigma}_{n-1}}$, which yields 
\[ \limsup_{N \to \infty} \frac1N \left|  \ln \tr e^{-\beta(\sqrt{N} X +V_N -B)} -  \max \{  \ln \tr e^{-\beta(\sqrt{Na_n} X_{\pmb{\sigma}_1\pmb{\sigma}_2\cdots \pmb{\sigma}_n} + V_N^{(n)})}, \ln \tr e^{-\beta( V_N^{(n)} -B )}   \} \right| =0.    \]
Writing $V_N^{(n)} =: V_N^{(n-1)} + \sqrt{N a_{n-1}}X_{\pmb{\sigma}_1\pmb{\sigma}_2\cdots \pmb{\sigma}_{n-1}}$ and using again Theorem~\ref{prop:key}, we similarly have 
\[ \limsup_{N \to \infty} \frac1N \left|\ln \tr e^{-\beta( V_N^{(n)} -B )} - \max \{ \ln \tr e^{-\beta( V_N^{(n)} -B^{2,x_{n-1}} )}, \ln \tr e^{-\beta( V_N^{(n-1)} -B)}    \}   \right| = 0.    \]
Proceeding like this, we arrive after $n$ steps at \eqref{eq:keycor}.
 
\end{proof}
Theorem \ref{thm:qgrem}
and Corollary \ref{cor:key} in case $V_N =0$ look alike. However, in Theorem \ref{thm:qgrem} we further evaluate the trace and claim that the maximum in \eqref{eq:keycor} is attained at some endpoint $y_l$ of the concave hull $\bar{A}$. We postpone the remaining part of the proof of Theorem \ref{thm:qgrem} to Section~\ref{sec:proofGREM}.

Instead, we will extend Corollary \ref{cor:key} to CREM type potentials. To this end, we introduce a useful shorthand notation. If $X$ is a centered Gaussian vector with hierarchical distribution function $A$, we define for $0 \leq z \leq 1$  the centered Gaussian vector $X^{(z)}$ with hierarchical distribution function given by
\[ A^{(z)}(x) \coloneqq \begin{cases}
 A(x) & \text{ if } x \leq z, \\
 A(z) & \text{ else.}
\end{cases} \]
We are now ready to formulate 
\begin{theorem}\label{thm:general}
Let $X$ be a centered Gaussian vector of CREM-type with distribution function $A$. Then, we have almost sure convergence
\begin{equation}\label{eq:genthm}
\limsup_{N \to \infty}   \left|\frac1N \ln \tr e^{-\beta(\sqrt{N}X + V_N - B)} -  \sup_{0 \leq z \leq 1}\frac1N \ln \tr e^{-\beta(\sqrt{N}X^{(z)} + V_N - B^{2,z})} \right| = 0. 
\end{equation}
\end{theorem}

Our proof of Theorem \ref{thm:general} relies on an interpolation argument. We first adapt the classical arguments to our setting with a transversal magnetic field.
We fix some inverse temperature $\beta$ and random field $B$.
Let $X,Y$ be two independent centered Gaussian vectors on $\mathcal{Q}_N$, which are independent of $V_N$ as well. For $t \in [0,1]$ we set
the interpolated pressure $\Psi$,
\[ \Psi(t) = \frac1N \ln \left[\tr e^{-\beta(\sqrt{tN}X + \sqrt{(1-t)N}Y + V_N -B} 
\right] , 
\]
where by Lemma~\ref{lem:possemgr} we may assume without loss of generality that $ b_j \geq 0 $ for all $ j $. 
By standard Gaussian interpolation (see e.g. Lemma 1.3.1 in~\cite{Tal11}), we obtain
\[ \ee_{X,Y}\left[\Psi(1) - \Psi(0)  \right] 
= \frac12 \sum_{\pmb{\sigma},\pmb{\sigma^\prime}} \int_{0}^{1} (\ee[Y(\pmb{\sigma})Y(\pmb{\sigma^\prime})]-\ee[X(\pmb{\sigma})X(\pmb{\sigma^\prime})]  )  \ee_{X,Y}\left[ \frac{\partial^2 \Psi(t)}{\partial X_{\pmb{\sigma}} \partial X_{\pmb{\sigma^\prime}} }
+ \frac{\partial^2 \Psi(t)}{\partial Y_{\pmb{\sigma}} \partial Y_{\pmb{\sigma^\prime}} }
\right] dt ,
\]
where $\ee_{X,Y}$ denotes the expectation with respect to $X$ and $Y$. In general, $ \ee_{X,Y}[\Psi(t)]$ is still a random variable due to the randomness of $V_N$ and $B$.
The second partial derivatives of $\Psi(t)$ can be computed explicitly:
\[\frac{\partial^2 \Psi(t)}{\partial X_{\pmb{\sigma}} \partial X_{\pmb{\sigma}^\prime} }
+ \frac{\partial^2 \Psi(t)}{\partial Y_{\pmb{\sigma}} \partial Y_{\pmb{\sigma}^\prime} }  = - \beta^2 \frac{\langle \pmb{\sigma}| e^{H_t} | \pmb{\sigma} \rangle\langle \pmb{\sigma}^\prime| e^{H_t} | \pmb{\sigma}^\prime \rangle  }{(\tr e^{H_t})^2} + \beta^2 \int_{0}^{1} \frac{\langle \pmb{\sigma}| e^{sH_t} | \pmb{\sigma}^\prime \rangle\langle \pmb{\sigma}^\prime| e^{(1-s)H_t} | \pmb{\sigma} \rangle  }{\tr e^{H_t}} ds \]
with the abbreviation $H_t \coloneqq -\beta(\sqrt{tN}X + \sqrt{(1-t)N}Y + V_N - B)$. Since we assumed without loss of generality that $b_j \geq 0 $, the matrix elements $\langle \pmb{\sigma}| e^{H_t} | \pmb{\sigma}^\prime \rangle$ are nonnegative for any $\pmb{\sigma},\pmb{\sigma}^\prime$. Moreover, we know that
\[ \sum_{\pmb{\sigma},\pmb{\sigma}^\prime}\frac{\langle \pmb{\sigma}| e^{H_t} | \pmb{\sigma} \rangle\langle \pmb{\sigma}^\prime| e^{H_t} | \pmb{\sigma}^\prime \rangle  }{(\tr e^{H_t})^2} = \sum_{\pmb{\sigma},\pmb{\sigma}^\prime} \int_{0}^{1} \frac{\langle \pmb{\sigma}| e^{sH_t} | \pmb{\sigma}^\prime \rangle\langle \pmb{\sigma}^\prime| e^{(1-s)H_t} | \pmb{\sigma} \rangle  }{\tr e^{H_t}} ds = 1.  \]
We consequently arrive at the bound 
\[|\ee_{X,Y}\left[\Psi(1) - \Psi(0)  \right]| \leq \beta^2 \max_{\pmb{\sigma},\pmb{\sigma}^\prime}|\ee[X(\pmb{\sigma})X(\pmb{\sigma}^\prime)] -\ee[Y(\pmb{\sigma})Y(\pmb{\sigma}^\prime)]  |. \]
In case  $X$ and $Y$ are of CREM-type with distribution functions $A_X$ and $A_Y$, respectively, we conclude 
\begin{equation}\label{eq:interpol}
\left|\ee_{X,Y}\left[\Psi(1) - \Psi(0)  \right]\right| \leq \beta^2 \|A _X- A_Y\|_\infty.
\end{equation}
Analogously, we get  
\begin{equation}\label{eq:interpol2}
 \frac1N \left| \ee_{X,Y} \left[\ln \tr e^{-\beta(\sqrt{N}X^{(z)} + V_N - B^{2,z})} - \ln \tr e^{-\beta(\sqrt{N}Y^{(z)} + V_N - B^{2,z})} \right] \right| \leq \beta^2 \|A_X - A_Y\|_\infty
\end{equation}
for any $z \in [0,1]$. 
The bounds \eqref{eq:interpol} and \eqref{eq:interpol2} are our first main ingredients for the proof of Theorem \ref{thm:general}. We observe, however, that an interpolation argument only controls the expectation value with respect to the Gaussian variables. The following Gaussian concentration inequality is a convenient method to lift the convergence of expectation values to almost sure statements and vice versa.

\begin{proposition}\label{prop:conc2}
	Let $X$ be a Gaussian vector of CREM-type, $V_N$ a random potential, and $B$ a random transversal field, all  independent from each other. The corresponding pressure
	\[ \Phi_N(\beta) = \frac1N \ln \tr e^{-\beta(\sqrt{N}X+V_N + B)} . \]
	exhibits a  Gaussian concentration estimate, i.e., for any $t > 0$ and $N \in \nn$
	\begin{equation}
	\mathbb{P}_X\left(\left| \Phi_N(\beta)  - \ee_{X}\left[\Phi_N(\beta) \right]  \right|  > \frac{t \beta }{\sqrt{N} } \right) \leq  2 \,  \exp\left(- \frac{t^2}{4}\right).
	\end{equation}	
 The same bounds holds true for $ \Phi_N^{(z)}(\beta) = \frac1N \ln \tr e^{-\beta(\sqrt{N}X^{(z)}+V_N - B^{2,z})}$.
\end{proposition}

\begin{proof}
	Since the lexicographic overlap~\eqref{eq:overlap}, can only take values $ k/N $ with $ k = 0, 1, \dots, N $ for every fixed $ N \in \mathbb{N} $, 
	the CREM-type Gaussian vector $ X $ may be represented as a GREM-type distribution.
	\[  X(\pmb{\sigma}) = \sqrt{a_1} X_{\pmb{\sigma}_1} + \sqrt{a_2}X_{\pmb{\sigma}_1\pmb{\sigma}_2} + \cdots + \sqrt{a_n}X_{\pmb{\sigma}_1\pmb{\sigma}_2\cdots \pmb{\sigma}_n}  \]
	with independent standard Gaussian variables 
	$X_{\pmb{\sigma}_1},\ldots,X_{\pmb{\sigma}_1\pmb{\sigma}_2\cdots \pmb{\sigma}_n}$ and some $ n = n(N) $. 
	 We  calculate the the free energy's variation with respect to the i.i.d Gaussian variable $X_{\pmb{\sigma}_1,\ldots,\pmb{\sigma}_k}$
	$$\displaystyle - \frac{\partial \Phi_N(\beta)}{\partial X_{\pmb{\sigma}_1,\ldots,\pmb{\sigma}_k}}   = \frac{\beta \sqrt{a_k}}{\sqrt{N} \tr e^{-\beta(X+V_N - B) } } \sum_{\hat{\pmb{\sigma}}_k}   \langle \pmb{\sigma}_1\cdots\pmb{\sigma}_k\hat{\pmb{\sigma}}_k  | e^{-\beta(X+V_N - B)} | \pmb{\sigma}_1\cdots\pmb{\sigma}_k\hat{\pmb{\sigma}}_k  \rangle . $$
	Here, $\hat{\pmb{\sigma}}_k$ is an abbreviation for the remaining entries of the element $\pmb{\sigma} \in \mathcal{Q}_N$.
	Consequently, the square of the pressure's Lipschitz constant is bounded by 
	$$\displaystyle \sum_{k} \sum_{\pmb{\sigma}_1\cdots\pmb{\sigma}_k}  \left( \frac{\partial \Phi_N(\beta)}{\partial X_{\pmb{\sigma}_1,\ldots,\pmb{\sigma}_k}}\right)^2  \leq  \frac{\beta^2}{N},  $$
	where we used that the weights $a_k$ add up to one. 
	If we condition on $V_N$ and $ B $, the Gaussian concentration inequality for Lipschitz functions
	(see \cite[Thm.~1.3.4]{Tal11}) yields 
	\[	\mathbb{P}_X\left(\left| \Phi_N(\beta, B)  - \ee_{X}\left[\Phi_N(\beta) \right]  \right|  > \frac{t \, \beta }{\sqrt{N} } \bigg| \, V \right) \leq  2 \,  \exp\left(- \frac{t^2}{4}\right).  \]
	A similar argument,using the fact that the sum of the weights $a_k^{(z)}$ add up to at most one,  shows that we have the same concentration inequality for $\Phi^{(z)}_N(\beta)$.
\end{proof}

Let us remark that a Gaussian concentration estimate still holds true if the weights $(a_k)$ do not add up to one. Only the multiplicative constant in front of the exponential term changes. 
We move on to the proof of Theorem \ref{thm:general}:
\begin{proof}[Proof of Theorem \ref{thm:general}]
	We pick some $\epsilon > 0$ and an independent  Gaussian vector $Y$ of GREM-type with distribution (step-)function $\tilde{A}$ such that 
	\[ \|A - \tilde{A} \|_{\infty} \leq \epsilon.  \]
	We can always find such a Gaussian vector as $A$ is a increasing, right-continuous function and, therefore a uniform limit of increasing step functions. We further note that this implies $\|A^{(z)} - \tilde{A}^{(z)} \|_{\infty} \leq \epsilon$. We denote by $0 = x_0 <x_1 < x_2 < \cdots x_n = 1$ the points, supporting $\tilde{A}$.
	
	We first exploit the estimates in \eqref{eq:interpol},\eqref{eq:interpol2} and Proposition \ref{prop:conc2} in order to obtain the almost sure bounds 
	\[ \limsup_{N \to \infty} \frac1N| \ln \tr e^{-\beta(\sqrt{N}X + V_N - B)} - \ln \tr e^{-\beta(\sqrt{N}Y + V_N - B)}| \leq \beta^2 \epsilon  \]
	and 
	\[ \limsup_{N \to \infty}\sup_{z \in [0,1] }  \frac1N| \ln \tr e^{-\beta(\sqrt{N}X^{(z)} + V_N - B^{2,z})} - \ln \tr e^{-\beta(\sqrt{N}Y^{(z)} + V_N - B^{2,z})}| \leq \beta^2 \epsilon . \]
	The expressions depending on $Y$ do not necessarily converge. Nevertheless, we have almost surely
	\[  \begin{split} & \hspace{0.38cm}  \limsup_{N \to \infty} \frac1N |\ln \tr e^{-\beta(\sqrt{N}Y + V_N - B)} - \sup_{0 \leq z \leq 1}\ln \tr e^{-\beta(\sqrt{N}Y^{(z)} + V_N - B^{2,z})}| \\ &= \limsup_{N \to \infty} \frac1N |\ln \tr e^{-\beta(\sqrt{N}Y + V_N - B)} - \max_{k = 0,1\ldots,n} \ln \tr e^{-\beta(\sqrt{N}Y^{(x_k)} + V_N - B^{2,x_k})}| = 0. \end{split}\]
	For the first equality we recall that for any $x_k \leq z < x_{k+1}$ the processes $Y^{(z)} = Y^{(x_k)}$. Consequently, the Gibbs' variational principle (with $ H = H' - (B^{2,x_k} -B^{2,z} ) $ and $ H' = \sqrt{N}Y^{(x_k)} + V_N -  B^{2,z} $ and an argument similar to \eqref{eq:Gibbs1}--\eqref{eq:Gibbs2}) shows that the maximum is attained at some $x_k$. The second equality follows from  Corollary~\ref{cor:key}. 
	Combining all these estimates, we arrive at
	\[ \limsup_{N \to \infty} \frac1N | \ln \tr e^{-\beta(\sqrt{N}X + V_N- B)} - \sup_{0 \leq z \leq 1} \ln \tr e^{-\beta(\sqrt{N}X^{(z)} + V_N -  B^{2,z})}| \leq 2 \beta^2 \epsilon. \]
	As $\epsilon >0$ is arbitrary, the proof of~\eqref{eq:genthm} is completed.
\end{proof}

 \section{Proofs of the main results}\label{sec:proofmain}
 \subsection{The Quantum GREM and CREM}\label{sec:proofGREM}
 We first aim to prove Theorem \ref{thm:qgrem}, i.e.
 \[\lim_{N \to \infty} \frac1N \ln \tr e^{-\beta(\sqrt{N} X - B)} = \max_{1 \leq l \leq m}\left[ \sum_{j =1}^{l} \varphi^{(j)}(\beta) + (1-y_l) \ee[\ln 2 \cosh(\beta \mathfrak{b})] \right]\]
 for a GREM type variable $X$ and transversal field $B$ consisting of independent weights $(b_j)$ with the same distribution as $\mathfrak{b}$. We recall that $x_1,\ldots,x_n$ denote the jump points of the distribution function $A$, the points $y_1,\ldots,y_m$, over which the above maximum is taken, are the endpoints of the concave hull's $\bar{A}$ linear segments and $\varphi^{(j)}(\beta)$ are the partial free energies from \eqref{eq:partfren}.
 For the remainder of this subsection and since we are interested in the limit $ N \to \infty $, we also assume without loss of generality that $ x_k N \in \mathbb{N} $ for all $ k \in \{ 0,\ldots, n \} $.

  Our starting point is Corollary \ref{cor:key} which for any GREM-type vector $X$ and $ V_N = 0 $ yields, 
     \begin{equation}\label{eq:pregrem} \frac1N \ln \tr e^{-\beta(\sqrt{N} X + B)}  =  \max_{k = 0,\ldots, n} \left[ \frac1N \ln \tr_{| \mathcal{Q}_{x_k N}} e^{-\beta \sqrt{N} X^{(x_k)}} + \frac{(1-x_k)}{N}  \sum_{i=1}^{N}\ln 2\cosh(\beta b_i) \right] + o(1) .
  \end{equation}
The limit $ N \to \infty $ of the bracket on the right side exists for any $ k \in \{ 0,\ldots, n \}$. More precisely, the strong law of large numbers implies that second term almost surely tends to $(1-x_k)\ee[\ln 2 \cosh(\beta \mathfrak{b})]$. Moreover, the first term converges since $X^{(x_k)}$ is still a GREM-type Gaussian vector on $\mathcal{Q}_{x_k N}$. The only difference is that the weights $a_1,\ldots a_k$ do not add up to $1$. This minor obstacle can be easily done away with rescaling the inverse temperature $\beta$. In particular, if $x_k$ coincides with an endpoint $y_l$ of the concave hull's segments, 
 \[ \lim_{N \to \infty} \frac1N \ln \tr_{| \mathcal{Q}_{x_k N}} e^{-\beta \sqrt{N} X^{(x_k)}} = \sum_{j = 1}^{l} \varphi^{(j)}(\beta), \]
 where the partial free energies $\varphi^{(j)}(\beta)$ remain unchanged, i.e., they are still given by \eqref{eq:partfren}. This follows from the observation that 
 $X$ and $X^{(y_l)}$ have the same concave hull up to the point $y_l$. 
  
Since the limit $ N\to \infty $ exits for each $ k $, we may exchange the limit with the maximum.  In order to prove Theorem \ref{thm:qgrem}
  it therefore suffices to check that in 
   \begin{equation}\label{eq:pgrem2} 
    \lim_{N \to \infty} \frac1N \ln \tr e^{-\beta(\sqrt{N} X + B)} = \max_{k = 0,\ldots, n}\left[ \lim_{N \to \infty} \frac1N \ln \tr_{| \mathcal{Q}_{x_k N}} e^{-\beta \sqrt{N} X^{(x_k)}} + (1-x_k)\ee[\ln 2 \cosh(\beta \mathfrak{b})] \right]
   \end{equation} 
   the maximum of the limit is always attained at some $y_l$. This is the content of the next Lemma:
  \begin{lemma}\label{lem:x=y}
  	If $X$ is a Gaussian vector of GREM type, we have 
  	\begin{equation}
  	\begin{split}
  	& \max_{k = 0,\ldots, n} \left[ \lim_{N \to \infty} \frac1N \ln \tr_{| \mathcal{Q}_{x_k N}} e^{-\beta \sqrt{N} X^{(x_k)}} + (1-x_k) \ee[\ln 2 \cosh(\beta \mathfrak{b})] \right] \\ & = \max_{l = 0,\ldots,m} \left[\lim_{N \to \infty} \frac1N \ln \tr_{| \mathcal{Q}_{y_l N}} e^{-\beta \sqrt{N} X^{(y_l)}} + (1-y_l)\ee[\ln 2 \cosh(\beta \mathfrak{b})]\right] .
  	 \end{split}
  	\end{equation}
  \end{lemma}

\begin{proof}
If $\{ x_0,\ldots, x_n \} = \{ y_0,\ldots, y_m \}$, the statement is trivial. So, let $y_l < x_k < y_{l+1}$. We recall that distribution function $A^{(x_k)}$ of $X^{(x_k)}$ is given by
\[ A^{(x_k)} = \begin{cases}
A(x) & \text{ if } x \leq x_k, \\
A(x_k) & \text{ else.}
\end{cases}.  \]
We introduce the Gaussian processes $Y$ and $Z$ of GREM type with the distribution functions 
\[ A_Y(x) \coloneqq \begin{cases}
A(x) & \text{ if } x \leq y_l, \\
A(y_l) & \text{ if } y_l < x < x_k \\
A(x_k) & \text{ if } x \geq x_k
\end{cases}  \]  
and, respectively, 
\[ A_Z(x) \coloneqq \begin{cases}
A(x) & \text{ if } x \leq y_l, \\
A(y_l) & \text{ if } y_l < x < x_k \\
A(y_l) + \frac{x_k - y_l}{y_{l+1} - y_l} (A(y_{l+1}) - A(y_l) & \text{ if } x \geq x_k.
\end{cases}  \] 
We note that
\[ \lim_{N \to \infty} \frac1N \ln \tr_{| \mathcal{Q}_{x_k N}} e^{-\beta \sqrt{N} X^{(x_k)}} \leq \lim_{N \to \infty} \frac1N \ln \tr_{| \mathcal{Q}_{x_k N}} e^{-\beta \sqrt{N} Y} \leq \lim_{N \to \infty} \frac1N \ln \tr_{| \mathcal{Q}_{x_k N}} e^{-\beta \sqrt{N} Z}. \] 
Here, the first inequality follows from Slepian's lemma (the less correlated a classical system is the higher is the pressure). For the second inequality we recall that $A$ is majorized by its concave hull $\bar{A}$ and agrees with $\bar{A}$ at $y_l$ and $y_{l+1}$, i.e., 
\[ A(x_k) \leq A(y_l) + \frac{x_k - y_l}{y_{l+1} - y_l} \left(A(y_{l+1}) - A(y_l) \right).  \]
Since the pressure is an increasing function of the jump heights (cf.~\eqref{eq:partfren}), we arrive at the second bound. The free energy of $Z$, is computed easily in terms of the partial pressure \eqref{eq:partfren} corresponding to $ A $:
\[ \lim_{N \to \infty} \frac1N \ln \tr_{| \mathcal{Q}_{x_k N}} e^{-\beta \sqrt{N} Z} = \sum_{j = 1}^{l} \varphi^{(j)}(\beta) + \frac{x_k-y_l}{y_{l+1} - y_l} \varphi^{(j+1)}(\beta). \]
We thus obtain 
\[\begin{split} & \lim_{N \to \infty} \frac1N \ln \tr_{| \mathcal{Q}_{x_k N}} e^{-\beta \sqrt{N} X^{(x_k)}} + (1-x_k) \ee[\ln 2 \cosh(\beta \mathfrak{b})] \\ &\leq \sum_{j = 1}^{l} \varphi^{(j)}(\beta) + (1-y_{l}) \ee[\ln 2 \cosh(\beta \mathfrak{b})] + \frac{x_k-y_l}{y_{l+1} - y_l} \left( \varphi^{(l+1)}(\beta) - (y_{l+1} - y_l) \ee[\ln 2 \cosh(\beta \mathfrak{b})] \right). \end{split} \]
Depending on the sign of the term in the last bracket we have 
\[  \lim_{N \to \infty} \frac1N \ln \tr_{| \mathcal{Q}_{x_k N}} e^{-\beta \sqrt{N} X^{(x_k)}} + (1-x_k) \ee[\ln 2 \cosh(\beta \mathfrak{b})] \leq \sum_{j = 1}^{l} \varphi^{(j)}(\beta) + (1-y_{l}) \ee[\ln 2 \cosh(\beta \mathfrak{b})]  \]
or 
\[  \lim_{N \to \infty} \frac1N \ln \tr_{| \mathcal{Q}_{x_k N}} e^{-\beta \sqrt{N} X^{(x_k)}} + (1-x_k) \ee[\ln 2 \cosh(\beta \mathfrak{b})]\leq \sum_{j = 1}^{l+1} \varphi^{(j)}(\beta) + (1-y_{l+1})\ee[\ln 2 \cosh(\beta \mathfrak{b})].  \]
Consequently, the maximal pressure is indeed attained at some $y_l$.
\end{proof}
The following observation is useful for the proof of Corollary 
\ref{cor:qgrem}:
\begin{lemma}\label{lem:concav}
	Let $\varphi^{(j)}(\beta)$ be partial pressure~\eqref{eq:partfren} and $L_j \coloneqq y_j - y_{j-1}$ the interval lengths. Then, the discrete concavity estimate
	\begin{equation}\label{eq:concav}
	\frac{\varphi^{(1)}(\beta)}{L_1} > \frac{\varphi^{(2)}(\beta)}{L_2} > \cdots > 	\frac{\varphi^{(m)}(\beta)}{L_m}
	\end{equation}
	holds for any inverse temperature $\beta > 0$.
\end{lemma}
\begin{proof}
We call $\varphi^{(j)}(\beta)$ "frozen" if $\beta > \beta_j$, i.e., $\varphi^{(j)}	(\beta)$ is given by the linear expression in~\eqref{eq:partfren}. 
Otherwise we say $\varphi^{(j)}(\beta)$ is "unfrozen".
By construction of the concave hull $\bar{A}$ we know that the slopes $\gamma_j = \bar{a}_j/L_j$ are strictly decreasing in $j$. The inequalities in \eqref{eq:concav}, where two consecutive partial free energies are either both frozen or both unfrozen, are thus obvious. It remains to consider the case where $\varphi^{(j)}(\beta)$ is frozen, but  $\varphi^{(j+1)}(\beta)$ is unfrozen. By \eqref{eq:partfren} we then have
\[ \frac{\varphi^{(j)}(\beta)}{L_j}  = \beta \sqrt{(2 \ln 2) \gamma_j}   \quad \mbox{and}\quad   \frac{\varphi^{(j+1)}(\beta)}{L_{j+1}}  = \frac{\beta^2}{2}  \gamma_{j+1} + \ln 2.  \]
Moreover, as $\varphi^{(j)}(\beta)$ is frozen and   $\varphi^{(j+1)}(\beta)$ is unfrozen, the inverse temperature satisfies
\[\beta_j=  \sqrt{(2 \ln 2) \gamma_j^{-1}} < \beta \leq \sqrt{(2 \ln2 )\gamma_{j+1}^{-1}} =\beta_{j+1}.  \]
We conclude that 
\[ \frac{\varphi^{(j)}(\beta)}{L_j} =  \frac{\beta}{\beta_j} 2 \ln 2 >  2 \ln 2 \geq \ln 2 + \frac{\beta^2}{\beta_{j+1}^2} \ln2 =  \ln 2 + \frac{\beta^2}{2} \gamma_{j+1} =\frac{\varphi^{(j+1)}(\beta)}{L_{j+1}}.    \]
\end{proof}

\begin{proof}[Proof of Theorem \ref{thm:qgrem} and Corollary \ref{cor:qgrem}] Theorem \ref{thm:qgrem} is an immediate consequence of~\eqref{eq:pgrem2}  and Lemma \ref{lem:x=y}. 

It remains to show Corollary \ref{cor:qgrem}. To this end, let us introduce the energy differences 
\[ \Delta^{(j)}(\beta,\Gamma) \coloneqq (y_j - y_{j-1}) \ln 2 \cosh(\beta \Gamma) - \varphi^{(j)}(\beta). \]
In view of Lemma \ref{lem:concav}, we conclude:
\begin{itemize}
	\item if $\Delta^{(j)}(\beta,\Gamma) < 0$ for some $j \geq 1$, then $\Delta^{(i)}(\beta,\Gamma) < 0$ for all $0 < i \leq j$
		\item if $\Delta^{(j)}(\beta,\Gamma) \geq 0$ for some $j \geq 1$, then $\Delta^{(i)}(\beta,\Gamma) \geq 0$ for all $j \leq i \leq m$
\end{itemize}
Consequently, the maximum in \eqref{eq:qgrem} is attained at $m$ if all energy differences $\Delta^{(j)}$ are negative for $0< j \leq m $ and, otherwise at the minimal integer  $K< m $ such that $\Delta^{(K+1)} \geq 0$. We may thus rewrite the pressure as 
\[ \Phi(\beta,\Gamma) = \sum_{l=1}^{m} \left( \varphi^{(l)}(\beta) \mathbbm{1}_{\Delta^{(l)} \leq 0} + L_l \ln 2 \cosh(\beta \Gamma)   \mathbbm{1}_{\Delta^{(l)} > 0}\right). \]
We note that the condition $\Delta^{(l)} > 0$ 
is equivalent to $\Gamma > \Gamma^{(l)}_c(\beta)$. This concludes the proof of~\eqref{eq:grem2}.	
\end{proof}

Our next goal is to prove Theorem \ref{thm:qcrem} and Corollary \ref{cor:qcrem}. It is convenient to use Theorem \ref{thm:qgrem} and the interpolation estimate \eqref{eq:interpol} rather than the general Theorem \ref{thm:general}. To do so, we first establish some continuity properties of the functions
\[ \Phi(\beta,A,z) := \sqrt{2 \ln 2} \beta \int_{0}^{\min\{x(\beta),z\}}
\sqrt{\bar{a}(x)}\, dx  + \mathbbm{1}_{z > x(\beta)}\left( \frac{\beta^2}{2}(\bar{A}(z)-\bar{A}(x(\beta))) + \ln 2(z- x(\beta)\right) \]
with respect to the distribution function $A$. Therefore, we emphasize here the dependence on $A$ in notation.

\begin{lemma}\label{lem:formcont}
	Let $A$ and $(A_n )_{n \in \nn}$ be distribution functions on $[0,1]$ such that $A_n$ converges uniformly to $A$ as $n \to \infty$. Then:
	\begin{itemize}
		\item[(i)] The concave hulls $\bar{A}_n$ converge uniformly to $\bar{A}$ as $n \to \infty$, i.e.,
		\[ \lim_{n \to \infty} \|\bar{A} - \bar{A}_n  \|_{\infty} = 0. \]
		\item[(ii)] The right derivatives $\bar{a}_n(x)$ converge to $\bar{a}(x)$ at any $x$, where $\bar{a}$ is continuous.
		\item[(iii)] For any $\beta \geq 0$ the functions $\Phi(\beta, A_n,z)$ converge uniformly to $\Phi(\beta, A,z)$ as a function of $z$, i.e.,
			\[ \lim_{n \to \infty} \|\Phi(\beta, A_n,\cdot) - \Phi(\beta, A,\cdot) \|_{\infty} = 0. \]
	\end{itemize} 
\end{lemma} 

\begin{proof}
	\begin{itemize}
		\item[1.] The function $\bar{A} + \|\bar{A} - \bar{A}_n  \|_{\infty}$ is a concave function which majorizes $A_n$, i.e.
		\[ \bar{A}_n \leq \bar{A} + \|\bar{A} - \bar{A}_n  \|_{\infty}.  \]
		Similarly, one shows that 
		\[\bar{A} \leq \bar{A}_n + \|\bar{A} - \bar{A}_n  \|_{\infty}.  \]
		The first assertion is a direct consequence of these bounds.
		\item[2.] Since $\bar{A}_n$ is a sequence of concave functions converging uniformly to $\bar{A}$, the second claim follows from standard convex analysis (see e.g.~\cite{Sim11}).
		\item[3.] We first recall that $x(\beta,A)$ is given by
		\[ x(\beta,A) = \sup \{x|\bar{a}(x) > 2 \ln2/ \beta^2 \}. \]
		Since $\bar{a}$ is a decreasing function,  $\bar{a}$ is continuous except for an at most countable set. The second statement implies then that $x(\beta,A_n)$ converges to $x(\beta,A)$. Next, we rewrite 
		\[ \Phi(\beta, A,z) = \int_{0}^{z} \varphi(\beta, A,x) \, dx \]
		with the function 
		\[ \varphi(\beta, A,x) \coloneqq \beta \sqrt{(2 \ln2 )\bar{a}(x)}  \mathbbm{1}_{x < x(\beta,A)} + 
		\left(\frac{\beta^2}{2} \bar{a}(x) + \ln2\right)  \mathbbm{1}_{x \geq x(\beta)}.  \]
		Therefore, it suffices to show 
		\[ \lim_{n \to \infty}  \int_{0}^{1} |\varphi(\beta, A,x) - \varphi(\beta,A_n,x)| \, dx = 0. \]
		Due to our previous considerations, we know that $\varphi(\beta,A_n,x)$ converges almost everywhere (with respect to the Lebesgue measure  and $x$) to $\varphi(\beta,A_n,x)$. Moreover, the functions $\varphi(\beta,A_n,x)$ are uniformly bounded at $[\delta, 1]$ for any $\delta > 0$, since $\varphi(\beta,A_n,x)$ is  decreasing in $x$. By dominated convergence we obtain 
		\[ \lim_{n \to \infty}  \int_{\delta}^{1} |\varphi(\beta, A,x) - \varphi(\beta,A_n,x)| \, dx = 0  \]
		for any $\delta > 0$. On the other hand,
		\[  \int_{0}^{\delta} |\varphi(\beta, A,x) - \varphi(\beta,A_n,x)| \, dx \leq \int_{0}^{\delta} \varphi(\beta, A,x) + \varphi(\beta,A_n,x) \, dx \leq \frac{\beta^2}{2}\left(\bar{A}(\delta) + \bar{A}_n(\delta) \right) + 2 \delta \ln2. \]
		As $\bar{A}$ is continuous, $\bar{A}(0) = 0$ and the sequence $\bar{A}_n$ converges uniformly, the third assertion follows as $\delta \to 0$.
	\end{itemize}
\end{proof}

We are now ready to show Theorem \ref{thm:qcrem} and Corollary 
\ref{cor:qcrem}.
\begin{proof}[Proof of Theorem \ref{thm:qcrem} and Corollary~\ref{cor:qcrem} ]
We pick a sequence of step functions $A_n$, which are also distribution functions and converge uniformly to the distribution function $A$. By Theorem \ref{thm:general} the expression for $\Phi(\beta,\mathfrak{b},A_n)$ may be rewritten as 
	\[\Phi(\beta,\mathfrak{b},A_n) = \sup_{0 \leq z \leq 1}\left[ \Phi(\beta,A_n,z) + (1-z) \ee[ \ln 2 \cosh(\beta \mathfrak{b})] \right].    \]
	By the interpolation estimate \eqref{eq:interpol} the left side converges to  the corresponding limit of the quantum CREM's pressure $\Phi(\beta,\mathfrak{b},A)$, whereas the right side converges to
	\[ \lim_{n \to \infty} \sup_{0 \leq z \leq 1} \Phi(\beta,A_n,z) + (1-z)\ee[ \ln 2 \cosh(\beta \mathfrak{b})] = \sup_{0 \leq z \leq 1}\left[  \Phi(\beta,A,z) + (1-z) \ee[ \ln 2 \cosh(\beta \mathfrak{b})]  \right] \] 
	by Lemma \ref{lem:formcont}. This completes the proof of Theorem \ref{thm:qcrem}.
	
	In case $\bar{A}$ is continuously differentiable and $\mathfrak{b} = \Gamma$, the convex  function $[0,1] \ni z \mapsto \Phi(\beta,A,z) + (1-z)  \ln\left( 2 \cosh(\beta \Gamma) \right)$ possesses a maximum in the interior of its domain if and only if there exists a solution $z \in (0,1)$ of 
	\[ \frac{ \partial \Phi(\beta,A,z)}{\partial z} - \ln 2 \cosh(\beta \Gamma) = 0. \]
	Otherwise the maximum is attained at $z = 0$ or $z = 1$. A straightforward calculation then leads to the formula in Corollary \ref{cor:qcrem}.	
\end{proof}

\subsection{The nonhierarchical GREM in a transversal field}\label{sec:nonhier}
We start with the proof of Theorem~\ref{thm:nonhgrem}. In the following we will use the notation introduced in Section 1.3.
\begin{proof}[Proof of Theorem~\ref{thm:nonhgrem} ] 
	Our strategy is to adapt the proof of Corollary~\ref{cor:key}.
	To be more precise, we introduce for any subset 
	 $J \in \mathcal{P}$ 
the restriction $B^{J}$  of $B$ to the subgraph spanned by the spins $\pmb{\sigma}_J$, that is, 
	\[ B^{J} \coloneqq \sum_{k \in J} B^{(k)},  \quad  B^{(k)} \coloneqq B^{1,x_k} - B^{1,x_{k-1}},  \] 
	and $B^{\emptyset} = 0$.
	We claim that for any two subsets $I,J \in \mathcal{P}$
	and any potential $V_N$ independent of $X_{\pmb{\sigma}_I}^{I}$, we have
	\begin{equation}\label{eq:keynonh}
	\limsup_{N \to \infty} \frac1N \left| \ln \tr e^{-\beta(\sqrt{Na_I} X_{\pmb{\sigma}_I}^{I} + V_N-B^J)} - \max\{\ln \tr e^{-\beta(V_N-B^{J})}, \ln \tr e^{-\beta(\sqrt{Na_I} X_{\pmb{\sigma}_I}^{I} + V_N-B^{J\setminus I}}   \}\right| = 0.   
	\end{equation}
	We note that $B^{J}$ can be represented as a transversal magnetic field whose weights corresponding to the complement $J^c$ are set zero. Thus, \eqref{eq:keynonh} follows from Theorem~\ref{prop:key} after possibly rearranging the spin components. 
	Using  \eqref{eq:keynonh} successively for each subset $J \in \mathcal{P}$ (where the remaining potential  $V_N$ might change from step to step), we finally arrive at
    \[ \limsup_{N \to \infty} \frac1N \left| \ln \tr e^{-\beta(V-B)} - \max_{\mathcal{F} \subset \mathcal{P}, D \in \mathcal{P}, D^c \cap \mathcal{F} = \emptyset }  \ln \tr e^{-\beta(\sum_{F \in \mathcal{F}} \sqrt{a_F N} X_{\pmb{\sigma}_F}^{F} - B^{D^c} )} \right| = 0     \]
   	where $ D^c \cap \mathcal{F} = \emptyset$ is understood elementwise, that is, $ D^c \cap F = \emptyset$ for any $F \in \mathcal{F}$. We note that the convexity of the exponential and the variables $X_{\pmb{\sigma}_J}^J$ being centered Gaussians, implies  (e.g. by \eqref{eq:keynonh}) for any $D^c \cap \mathcal{F} = \emptyset,$
    	\[ \liminf_{N \to \infty}  \frac1N \left(  \ln \tr e^{-\beta(\sum_{F \in \mathcal{P}(D)} \sqrt{a_F N} X_{\pmb{\sigma}_F}^{F} - B^{D^c} )} - \ln \tr e^{-\beta(\sum_{F \in \mathcal{F}} \sqrt{a_F N} X_{\pmb{\sigma}_F}^{F} - B^{D^c} )} \right) \geq 0,  \]
    	where $\mathcal{P}(D)$ is the power set of $D$.
    	On the other hand, the limit of the left term exists almost surely and is given by 
    	\[ \lim_{N \to \infty} \frac1N \ln \tr e^{-\beta(\sum_{F \in \mathcal{P}(D)} \sqrt{a_F N} X_{\pmb{\sigma}_F}^{F} - B^{D^c} )} = \min_{S \in \mathcal{C}^{D}} \Phi_D(\beta,S)   + \sum_{k \in D^c} L_k \ \ee[\ln 2\cosh(\beta \mathfrak{b})],    \]
    	where we used the strong law of large numbers for the expression involving $B^{D^c}$ and the known convergence \cite{BK06} of the classical nonhierarchical GREM. We in fact need a slightly more generalized version of  \cite{BK06} which is also applicable to the reduced model on the subgraph generated by $\pmb{\sigma}_{D^c}$. However, this easily follows by a scaling argument. 
    Combining our findings, we arrive at
	\[ \lim_{N \to \infty} \Phi_N(\beta,\Gamma) = \max_{D \in \mathcal{P}} \min_{S \in \mathcal{C}^{D}} \left[ \Phi_D(\beta,S)   + \sum_{k \in D^c} L_k\  \ee[\ln 2\cosh(\beta \mathfrak{b})] \right].  \] 
\end{proof}

It remains to show Corollary~\ref{cor:nonhgrem}. To this end we need the following Lemma for the classical non-hierarchical GREM:

\begin{lemma}\label{lem:nonhgrem}
Let $X$ be a Gaussian vector of non-hierarchical GREM type. Then, there exists a chain $S_0$ such that for any chain $S$ the pointwise estimate 
\begin{equation}\label{eq:lemnhg1}
\bar{A}_{S}(x) \leq \bar{A}_{S_0}(x)
\end{equation}
holds true, where $\bar{A}_{S}$ and $\bar{A}_{S_0}$ are the concave hulls of the ordinary GREM vectors assigned to $S$ and $S_0$ respectively. Moreover, we have for any $\beta \geq 0$
\begin{equation}\label{eq:lemnhg2}
\Phi(\beta) = \min_{S \in \mathcal{C}} \Phi(\beta,S) =  \Phi(\beta,S_0).
\end{equation}
\end{lemma}

We note that the second assertion in Lemma~\ref{lem:nonhgrem} states that Corollary~\ref{cor:nonhgrem} holds true in the case $\mathfrak{b} = 0$. The statements of  Lemma~\ref{lem:nonhgrem} are a simple consequence of the derivation in \cite{BK06} (cf.~Remark~7 in that paper). For  completeness, we will spell out a proof in the appendix.

\begin{proof}[Proof of Corollary~\ref{cor:nonhgrem}]
Let $S_0$ be a chain as in Lemma~\ref{lem:nonhgrem}. After relabeling the components of $\pmb{\sigma}$, we may assume that 
\begin{equation}\label{eq:S_0} S_0 = \{\emptyset,\{1\},\{1,2\},\ldots,\{1,\ldots,n\}    \}     \end{equation}
We want to show that 
\[ \max_{D \in \mathcal{P}}\left[  \min_{S \in \mathcal{C}^{D}} \Phi(\beta,S)   + \sum_{k \in D^c} L_k \ \ee[\ln 2\cosh(\beta \mathfrak{b})] \right]= \Phi(\beta,\mathfrak{b},S_0)  \]	
by establishing two inequalities. First, abbreviating $ D_k:= \{1,\ldots,k \}  $ with $D_0:= \emptyset $, we have 

\[ \begin{split} &\max_{D \in \mathcal{P}} \left[   \min_{S \in \mathcal{C}^{D}} \Phi(\beta,S)   + \sum_{k \in D^c} L_k \ee[\ln 2\cosh(\beta \mathfrak{b})] \right]\\ & \geq
\max_{0 \leq k \leq n }  \left[ 
\min_{S \in \mathcal{C}^{D_k}} \Phi(\beta,S)   + \sum_{k \in D_k^c}
 L_k\ \ee[\ln 2\cosh(\beta \mathfrak{b})] \right] \\
 & =  \max_{0 \leq k \leq n}   \left[  \Phi(\beta,S_0^{D_k})   + \sum_{k \in D_k^c} L_k \ \ee[\ln 2\cosh(\beta \mathfrak{b})] \right] = 
 	\Phi(\beta,\mathfrak{b},S_0)
  \end{split}. \]
Here $S_0^D$ denotes the chain which coincides with $S_0$ but ends at $D$. The last line follows from Lemma~\ref{lem:nonhgrem} as it implies that even in the constrained setting the cut versions of $S_0$ are indeed minimizing chains.

For the reverse inequality, let $I_1,\ldots,I_m$ be the sets associated to the concave hull $\bar{A}_{S_0}$ and $\varphi_{S_0}^{(l)}(\beta)$ be the partial pressure corresponding to the GREM assigned to the chain $S_0$, cf.~\eqref{eq:partfren}. Moreover, for any $D \in \mathcal{P}$ we define the ordered-restriction chain $S_0^D$,
\[ S_0^D \coloneqq \{ \{\emptyset \},\{j_1\},\{j_1,j_2  \},\ldots \{j_1,\ldots,j_{k_D}\} \},   \]  
where $j_1 < j_2 < \ldots j_{k_D} \in D$ and $\{j_1,\ldots,j_{k_D}\} = D$. Then for any $D \in \mathcal{P}$ ,
\[ \begin{split} &\min_{S \in \mathcal{C}^{D}} \Phi(\beta,S)   + \sum_{j \in D^{c}} L_j \ \ee[\ln 2\cosh(\beta \mathfrak{b})] 
\leq  \Phi(\beta,S_0^D)   + \sum_{j\in D^{c}} L_j \ \ee[\ln 2\cosh(\beta \mathfrak{b})] \\
&= \sum_{l=1}^{m_D}  \varphi^{(l)}_{S_0^D}(\beta)  +
\sum_{j \in D^{c}} L_j \ee[\ln 2\cosh(\beta \mathfrak{b})] \\
&\leq \sum_{l=1}^{m}  \frac{\sum_{k \in I_l \cap D} L_k}{\sum_{k \in I_l} L_k} \varphi^{(l)}_{S_0}(\beta)  +
\sum_{j \in D^{c}} L_j \ee[\ln 2\cosh(\beta \mathfrak{b})].
  \end{split}   \]
  The last inequality, follows from three observations. First, we recall that the weights $a_{l}^{S_0^D}$ assigned to the chain $S_0^D$ are less or equal to the weights  $a_{j_l}^{S_0}$ of the chain~\eqref{eq:S_0}. Secondly, we note that the increments $\Delta_l \bar{A}_{S_0^D}$ on the segments $D \cap I_l \neq \emptyset $ can be bounded,
  \[\frac{ \Delta_l \bar{A}_{S_0^D}}{\sum_{k \in I_l \cap D} L_k}  \leq \frac{\Delta_l \bar{A}_{S_0}}{\sum_{k \in I_l} L_k} ,  \]
  since otherwise if this  does not hold  we construct a chain $S^\prime$ violating Lemma~\ref{lem:nonhgrem} using the first observation. Thirdly, an application of Slepian's lemma as in the proof of Lemma~\ref{lem:x=y} extends the summation to $ m $ and yields the claimed inequality.

We thus obtain 
  \[ \begin{split} & \sum_{l=1}^{m} \left[ \frac{\sum_{k \in I_l \cap D} L_k}{\sum_{k \in I_l} L_k} \varphi^{(l)}_{S_0}(\beta)  +
\sum_{j \in D^{c}\cap I_l} L_j \ee[\ln 2\cosh(\beta \mathfrak{b})] \right] \\
& \leq \sum_{l=1}^{m} \sum_{k \in I_l } L_k \max\left\{\frac{\varphi^{(l)}_{S_0}(\beta)}{\sum_{k \in I_l} L_k}    ,  \ee[\ln 2\cosh(\beta \mathfrak{b})]  \right\}  = \Phi(\beta,\mathfrak{b},S_0),  
  \end{split}   \]
 where the last equality is based on the concavity Lemma~\ref{lem:concav} and the explicit expression~\eqref{eq:qgrem} for the pressure of the quantum GREM. This completes the proof as $D$ was chosen arbitrarily.
 
\end{proof}

 \appendix
 \section{Supplementary results}

\subsection{Sufficient condition for Assumption~\ref{ass1}}

We want to present a quite general condition on the distribution of $X_{\pmb{\sigma_1}}$ which implies the third point in
Assumption~\ref{ass1}:

\begin{lemma}\label{lem:slln}
	Let $X_{\pmb{\sigma}_1}$ be independent and identically distributed centered random variables which satisfy an LDP with good rate function $I$, i.e., the sets $\{x| I(x) \leq a\}$ are compact for any $a \geq 0$. Moreover, the rate function shall satisfy 
	\[ \inf_{|x| > \varepsilon} I(x) > 0  \] 
	for any $\epsilon > 0$.
	Then, $(X_{\pmb{\sigma_1}})$ fulfills the conditions 1.,2. and 3. in Assumption~\ref{ass1}.
\end{lemma}

\begin{proof}
	The points 1. and 2. are clear and it remains to check 3.
	Let $w_{\pmb{\sigma}_1}$ be random weights which are independent from  $X_{\pmb{\sigma}_1}$ and satisfy almost surely $w_{\pmb{\sigma}_1} \geq 0$ and $\sum_{\pmb{\sigma}} w_{\pmb{\sigma}_1} =1$. We introduce the sets 
	\[ A_N \coloneqq \{\pmb{\sigma}_1 \in \mathcal{Q}_N^{(1)}|  w_{\pmb{\sigma}_1} \geq 1/N^2  \}   \]
	and show separately 
	\begin{equation}\label{eq:slln1}
	\lim_{N \to \infty} \frac1N \sum_{\pmb{\sigma}_1 \in A_N } w_{\pmb{\sigma}_1} X_{\pmb{\sigma_1}} = 0
	\end{equation}
	and 
	\begin{equation}\label{eq:slln2}
	\lim_{N \to \infty} \frac1N \sum_{\pmb{\sigma} \in A_N^c } w_{\pmb{\sigma}_1} X_{\pmb{\sigma_1}} = 0.
	\end{equation}
	\noindent
	\textit{Proof of \eqref{eq:slln1}}: We have the trivial bound
	\[ \left|\frac1N \sum_{\pmb{\sigma}_1 \in A_N } w_{\pmb{\sigma}_1} X_{\pmb{\sigma_1}} \right| \leq \frac1N \sup_{\pmb{\sigma}_1 \in A_N } |X_{\pmb{\sigma}_1}|.    \] 
	We note that the cardinality of $A_N$ is bounded by $N^2$
	so that the independence of $w_{\pmb{\sigma}_1}$ and $X_{\pmb{\sigma_1}}$ implies for any $\delta > 0$
	\[ \begin{split} \pp(\sup_{\pmb{\sigma}_1 \in A_N } |X_{\pmb{\sigma}_1}| &\geq \delta N) \leq N^2 \pp( |X_{\pmb{\sigma}_1}| \geq \delta N) \leq N^2 e^{-(c_\delta + o(1))N}, \\ c_\delta &= \inf_{|x| \geq \delta} I(x)  > 0.  \end{split}  \]
	Therefore, the bound on the probability is summable in $N$ for any $\delta > 0$ and a Borel-Cantelli argument finishes the proof of \eqref{eq:slln1}. \\
	
	\noindent
	\textit{Proof of \eqref{eq:slln2}}: As $I$ is a good rate function, we find  $C > 0$ such that 
	\[ \inf_{|x| \geq C} I(x) \geq 2 \ln 2,    \]
	and hence
	\[ \pp(\sup_{\pmb{\sigma}_1 \in \mathcal{Q}_N^{(1)} } |X_{\pmb{\sigma}_1}| \geq CN) \leq 2^N e^{-(2 \ln 2 + o(1))N} = (2+o(1))^{-N}.\]
	By a Borel-Cantelli argument 
	we may assume without loss of generality that $|X_{\pmb{\sigma}_1}| \leq CN$ holds true for all sufficiently large $N$ with probability one.  
	Then, by independence   we have
	\[ \begin{split}
	 \ee\left[ \left( \frac1N \sum_{\pmb{\sigma}_1 \in A_N^c } w_{\pmb{\sigma}_1} X_{\pmb{\sigma_1}} \right)^2\right] &\leq \frac{1}{N^2} \ee\left[  \sum_{\pmb{\sigma}_1 \in A_N^c } w_{\pmb{\sigma}_1}^2 X_{\pmb{\sigma_1}}^2 \right]
	\leq C^2 \ \ee\left[  \sum_{\pmb{\sigma}_1 \in A_N^c } w_{\pmb{\sigma}_1}^2 \right] \\ &\leq 
	\frac{C^2}{N^2} \ \ee\left[  \sum_{\pmb{\sigma}_1 \in A_N^c } w_{\pmb{\sigma}_1} \right]  \leq \frac{C^2}{N^2}. \end{split}   \]
	Here, the first bound follows from the independence of $ X_{\pmb{\sigma_1}}$ and $\ee[X_{\pmb{\sigma_1}}] =0 $ after conditioning on $w_{\pmb{\sigma}_1}$. Then, we use  $|X_{\pmb{\sigma}_1}| \leq CN$ and $w_{\pmb{\sigma}_1} \leq N^{-2}$ for $\pmb{\sigma}_1 \in A_N^c$. The Borel-Cantelli lemma again completes the proof.
	
\end{proof}

\subsection{Assumption~\ref{ass2} for independent $L^1$ weights}
The aim of this section is to verify that Assumption~\ref{ass2}
is satisfied for independent copies $(b_j)$ of an absolutely integrable variable $\mathfrak{b}$:

\begin{lemma}\label{lem:upb}
	If the weights  $b_i$ are independent copies of an absolutely integrable variable $\mathfrak{b}$, we almost surely have 
	\begin{equation}
	\limsup_{N \to \infty} N^{-1} \sqrt{\sum_{i=1}^{N} |b_i|^2} = 0.
	\end{equation}
\end{lemma}
\begin{proof}
	Our proof relies on a thinning and truncation argument and is similar to the proof of the strong law of large numbers in the $L^1$-case.
	
	Let us abbreviate the  partial sums $ S_N \coloneqq \sum_{i=1}^{N} |b_i|^2 $.
	We pick some $\epsilon > 0$ and introduce the sequence $N_m \coloneqq 2^m$. 
	Suppose we have already shown the almost sure convergence 
	\begin{equation}\label{eq:thin} \lim_{m \to \infty} (N_m)^{-2} S_{N_m} = 0.
	\end{equation} 
	Since $S_N$ is an increasing sequence we conclude that 
	\[ \limsup_{N \to \infty}  \frac{S_N}{N^2} \leq \limsup_{m \to \infty} \frac{S_{N_m}}{N_{m-1}^2} = 4  \limsup_{m \to \infty}  \frac{S_{N_m}}{N_{m}^2} = 0.  \]
	So it suffices to show \eqref{eq:thin}. To this end, let $K_m$ be a nonnegative  sequence which we will fix later and $S_{N_m}^{<}$,
	$S_{N_m}^{>}$ the truncated sums given by 
	\[ S_{N_m}^{<} \coloneqq \sum_{i =1}^{N_m} |b_i|^2 \mathbbm{1}_{|b_i| \leq K_m}  \quad \mbox{and} \quad S_{N_m}^{>} \coloneqq S_{N_m}  - S_{N_m}^{>}. \]
	For any $\epsilon > 0$ a Markov-type estimate yields 
	\[ \pp(S_{N_m}^{<} > \epsilon N_m^2) \leq \frac{\ee[|\mathfrak{b}|^2 \mathbbm{1}_{|\mathfrak{b}| \leq K_m} ]}{\epsilon N_m}. \]	
	We also have 
	\[ \pp(S_{N_m}^{>} \neq 0) \leq N_m \pp(|\mathfrak{b}| > K_m)  \]
	So, by a Borel-Cantelli argument the assertion follows if we can choose $K_m$ such that 
	\[ \sum_{m=1}^{\infty} \frac{\ee[|\mathfrak{b}|^2 \mathbbm{1}_{|\mathfrak{b}| \leq K_m} ]}{ N_m} + \sum_{m=1}^{\infty} N_m \pp(|\mathfrak{b}| > K_m) < \infty.    \]
	We claim that this can be accomplished by setting $N_m = K_m$. We note that the second sum is finite as $\mathfrak{b}$ is absolutely integrable.
	On the other hand,
	\[ \sum_{m=1}^{\infty} \frac{\ee[|\mathfrak{b}|^2 \mathbbm{1}_{|\mathfrak{b}| \leq N_m} ]}{ N_m} \leq 2 \sum_{m = 1}^{\infty} N_m^2 \pp(|\mathfrak{b}| \geq N_m) \sum_{k \geq m} N_m^{-1} \leq 4 \sum_{m = 1}^{\infty} N_m \pp(|\mathfrak{b}| \geq N_m) < \infty,  \]
	where the first inequality is a consequence of the layer-cake representation and the last bound is again a consequence of $\mathfrak{b}$ being absolutely integrable.
\end{proof}

\subsection{Proof of Lemma~\ref{lem:nonhgrem}}

Let us first recall Lemma~\ref{lem:nonhgrem}:
\begin{lemma}[= Lemma~\ref{lem:nonhgrem} ]
	Let $X$ be a Gaussian vector of nonhierarchical GREM type. Then, there exists a chain $S_0$ such that for any chain $S$ the pointwise estimate 
	\begin{equation}
	\bar{A}_{S}(x) \leq \bar{A}_{S_0}(x)
	\end{equation}
	holds true, where $\bar{A}_{S}$ and $\bar{A}_{S_0}$ are the concave hulls of the ordinary GREM vectors assigned to $S$ and $S_0$ respectively. Moreover, we have for any $\beta \geq 0$
	\begin{equation}
	\Phi(\beta) = \Phi(\beta,S).
	\end{equation}
\end{lemma}

\begin{proof}
	For any $\emptyset \neq J \subset \{1,\ldots,n   \}$ we define the corresponding slope $\gamma_J$,
	\[ \gamma_J \coloneqq \frac{\tilde{a}_J}{L_J} \coloneqq \frac{\sum_{I \subset J} a_I}{\sum_{k \in J} L_k}. \]
	We now construct a (possibly incomplete) chain $ J_1 \subset J_2 \subset \cdots J_m = \{1,\ldots,n\}  $
	as follows. We first pick a subset $J_1$ with maximal slope $\gamma_{J_1}$. If $J_1 = \{1,\ldots,n\}$, we are done. Otherwise we pick a subset $J_2$ such that 
	\[ \gamma_{J_2} = \max_{I \subset \{1,\ldots,n\}; I \not\subset J_1} a_I.   \]  
	One easily checks that $\gamma_{J_2} \leq \gamma_{J_1 \cup J_2}$, so we may assume that $J_1 \subset J_2$. We stop if $J_2 =  \{1,\ldots,n\}$ and continue the procedure otherwise. After at most $n$ steps we arrive at a (possibly incomplete) chain as claimed. We set $S_0$ to be a completion of $J_1,\ldots,J_m$, that is, $S_0$ is a chain which contains $J_1,\ldots,J_m$. Clearly, $S_0$ does not depend on $\beta$. 
	
	Both assertions follow now easily. We see that the concave hull  $\bar{A}_{S_0}$ assigned to $ S_0 $ is the unique piecewise linear function satisfying $\bar{A}_{S_0}(L_{J_k}) = \tilde{a}_{J_k}$ for any $k$. By construction, $\bar{A}_{S_0}$ is pointwise maximal as we iteratively pick the subset $J_k$ leading to the maximal mean slope. On the other hand, the bound $\bar{A}_{S_0} \geq \bar{A}_{S}$ for any chain $S$ 
	yields by Slepian's lemma
	$ \Phi(\beta,S_0) \leq \Phi(\beta,S)  $
	from which the second statement follows.
\end{proof}

 \minisec{Acknowledgments}
This work was supported by the DFG under EXC-2111 -- 390814868.

	\bigskip
	\bigskip
	\begin{minipage}{0.5\linewidth}
		\noindent Chokri Manai and Simone Warzel\\
		MCQST \& Zentrum Mathematik \\
		Technische Universit\"{a}t M\"{u}nchen
	\end{minipage}

\end{document}